\documentclass{LMCS}

\def\doi{8 (2:04) 2012}
\lmcsheading%
{\doi}
{1--39}
{}
{}
{Oct.~26, 2011}
{Apr.~27, 2012}
{}

\usepackage{hyperref}
\usepackage{proof}

\usepackage{graphicx}

\usepackage{amsmath}
\usepackage{amsthm}
\usepackage{amssymb}

\usepackage{ifthen}
\usepackage{graphicx}
\usepackage{fancybox}
\usepackage{boxedminipage}
\usepackage{enumerate}
\usepackage{alltt}
\usepackage{proof}
\usepackage{rotating,graphics}

\theoremstyle{plain}
\newtheorem{theorem}{Theorem}[section]
\newtheorem{lemma}[theorem]{Lemma}

\newtheorem{proposition}[theorem]{Proposition}

\theoremstyle{definition}
\newtheorem{definition}[theorem]{Definition}

\newtheorem{remark}[theorem]{Remark}

\theoremstyle{remark}

\usepackage{amsmath}
\usepackage{amstext}
\usepackage{amssymb}
\usepackage{amsfonts}
\usepackage{xspace}
\usepackage{graphicx}
\usepackage{url}
\usepackage{mdwlist}
\usepackage{courier}
\usepackage{lscape}
\usepackage{comment}
\usepackage{wrapfig}
\usepackage{multirow}
\usepackage{bussproofs}
\usepackage{pdftricks}
\begin{psinputs}
\usepackage{pstricks}
\usepackage{color}
\usepackage{pstcol}
\usepackage{pst-plot}
\usepackage{pst-tree}
\usepackage{pst-eps}
\usepackage{multido}
\usepackage{pst-node}
\usepackage{pst-eps}
\end{psinputs}
\usepackage{cite}
\makeatletter
\newif\if@restonecol
\makeatother

\usepackage[ruled,vlined]{algorithm2e}
\usepackage{footmisc}
\usepackage{balance}
\usepackage{enumerate}

\newcommand{\diff}{\ensuremath{\mathtt{diff}}\xspace}
\newcommand{\ARRAY}{\ensuremath{\mathtt{ARRAY}}\xspace}
\newcommand{\INDEX}{\ensuremath{\mathtt{INDEX}}\xspace}
\newcommand{\ELEM}{\ensuremath{\mathtt{ELEM}}\xspace}
\newcommand{\formulae}{formul\ae\xspace}

\newcommand{\AX}{\ensuremath{\mathcal{AX}}\xspace}
\newcommand{\AXEXT}{\ensuremath{\mathcal{AX}_{{\rm ext}}}\xspace}
\newcommand{\AXDIFF}{\ensuremath{\AX_{\diff}}\xspace}

\newcommand{\imp}{\Rightarrow}

\newcommand{\dpll}[1]{{\sc DPLL}\xspace}

\newcommand{\sep}{\ensuremath{\ | \ }}

\newcommand{\COMMENT}[1]{}

\newcommand{\alocal}{\ensuremath{A}-local\xspace}
\newcommand{\blocal}{\ensuremath{B}-local\xspace}

\newcommand{\abcommon}{\ensuremath{AB}-common\xspace}
\newcommand{\abpure}{\ensuremath{AB}-pure\xspace}
\newcommand{\abmixed}{\ensuremath{AB}-mixed\xspace}

\newcommand{\ua}{\ensuremath{\underline a}}

\newcommand{\ug}{\ensuremath{\underline g}}

\newcommand{\ux}{\ensuremath{\underline x}}

\newcommand{\cM}{\ensuremath \mathcal M}
\newcommand{\cN}{\ensuremath \mathcal N}

\newcommand{\uguale}{\ensuremath{=}\xspace}
\newcommand{\coincide}{\ensuremath{\equiv}\xspace}
\newcommand{\nform}{modular\xspace }

\begin{document}

\title{%
Quantifier-free Interpolation of a 
Theory of Arrays
}

\author[R.\ Bruttomesso]{Roberto Bruttomesso\rsuper a}	
\address{{\lsuper a}Dipartimento di Scienze dell'Informazione, Universit\`a degli Studi di Milano (Italy)}	
\email{bruttomesso@dsi.unimi.it}  

\author[S.\ Ghilardi]{Silvio Ghilardi\rsuper b}	
\address{{\lsuper b}Dipartimento di Matematica, Universit\`a degli Studi di Milano (Italy)}	
\email{ghilardi@dsi.unimi.it}  

\author[S.\ Ranise]{Silvio Ranise\rsuper c}	
\address{{\lsuper c}FBK-Irst, Trento (Italy)}	
\email{ranise@fbk.eu}  

\titlecomment{This paper is a substantially extended version
  of~\cite{rta11}.}

\begin{abstract}
  The use of interpolants in model checking is becoming an enabling
  technology to allow fast and robust verification of hardware and
  software. The application of encodings based on the theory of
  arrays, however, is limited by the impossibility of deriving
  quantifier-free interpolants in general.  

  In this paper, we show that it is possible to obtain quantifier-free
  interpolants for a Skolemized version of the extensional theory of
  arrays.  We prove this in two ways:
  \begin{enumerate}[(1)]
    \item non-constructively, by using the model theoretic notion of
      amalgamation, which is known to be equivalent to admit
      quantifier-free interpolation for universal theories; and
    \item constructively, by designing an interpolating procedure,
      based on solving equations between array updates.
      (Interestingly, rewriting techniques are used in the key steps
      of the solver and its proof of correctness.)
  \end{enumerate}
  To the best of our knowledge, this is the first successful attempt
  of computing quantifier-free interpolants for a variant of the
  theory of arrays with extensionality.
\end{abstract}

\maketitle

\newpage
\tableofcontents
\newpage

\section{Introduction}


Craig's interpolation theorem~\cite{Cra57} applies to first order
logic \formulae and states that whenever the sequent $A \wedge B \imp
\bot$ is valid, then it is possible
to derive a formula $I$ such that $(i)$ $A \imp I$ is valid , $(ii)$
$I \wedge B \imp \bot$ is valid, and $(iii)$ $I$ is defined over the
common symbols of $A$ and $B$.\footnote{To be precise, the original
  formulation of~\cite{Cra57} is slightly different, and it states
  that whenever $A \imp B$ is valid, then it is possible to derive an
  $I$ such that $A \imp I \imp B$ are valid, and $I$ is over the
  common symbols of $A$ and $B$.  Clearly, the two formulations are
  equivalent.}  After the seminal work of McMillan (see,
e.g.,~\cite{McM04a}), Craig's interpolation has become an important
technique in verification.  Intuitively, the interpolant $I$ can be
seen as an over-approximation of $A$ with respect to $B$.  This
observation is crucial for several applications of interpolation in
verification.  For example, the importance of computing
\emph{quantifier-free} interpolants (as several symbolic verification
procedures represent sets of states and transitions as quantifier-free
formulae) to over-approximate the set of reachable states for model
checking has been observed.  Unfortunately, Craig's interpolation
theorem does not guarantee that it is always possible to compute
quantifier-free interpolants.  Even worse, for certain first-order
theories, it is known that quantifiers must occur in interpolants of
quantifier-free formulae~\cite{KMZ06}.  As a consequence, several
papers~\cite{jhala,Kra97,McM04b,Pud97,stokkermans,YM05,KMZ06,RS07,CGS08,lynch,CGS10,BKR+10}
focused on the efficient computation of quantifier-free interpolants
for first-order theories which are relevant for verification such as
uninterpreted functions, (fragments of) Presburger arithmetic,
theories of some data-structures, and their combination.  Despite the
ongoing efforts, so far, only the negative result in~\cite{KMZ06} is
available for the computation of interpolants in the theory of arrays
with extensionality, axiomatized by the following three sentences:
\begin{eqnarray*}
  & & \forall y, i, e. rd(wr(y,i,e),i)\uguale e \\
  & & \forall y, i, j, e.  i \not\uguale j \imp rd(wr(y,i,e),j)\uguale rd(y,j) \\
  & & \forall x, y.  x \not\uguale y \imp (\exists i.\ rd(x,i)\not\uguale rd(y,i))
\end{eqnarray*}
where $rd$ and $wr$ are the usual operations for reading or updating
arrays, respectively.  For instance, there is no quantifier-free
interpolant for the pair of quantifier-free formulae
\begin{eqnarray*}
  \begin{array}{l}
  A \coincide  x\uguale wr(y,i,e) \\
  B \coincide rd(x,j)\not\uguale rd(y,j)\wedge rd(x,k)\not\uguale rd(y,k)\wedge j\not\uguale k . 
  \end{array}
\end{eqnarray*}
This theory is important for both hardware and software verification,
and a procedure for computing quantifier-free interpolants
``\emph{would extend the utility of interpolant extraction as a tool
  in the verifier's toolkit}''~\cite{McM04a}.  Indeed, the endeavour
of designing such a procedure would be bound to fail (according
to~\cite{KMZ06}) if we restrict ourselves to the original theory.  To
circumvent the problem, we add the (binary) function $\diff$ to $rd$
and $wr$.  Intuitively, $\diff(a,b)$ is an index at which the elements
stored in the arrays $a$ and $b$ are different ($\diff(a,b)$ is
defined arbitrarily in case $a$ and $b$ coincide).  Formally, this is
characterized by Skolemizing the third axiom above (also called the
extensionality axiom) to obtain
\begin{eqnarray*}
  \forall x, y.  x \not\uguale y \imp rd(x,\diff(x,y))\not\uguale rd(y,\diff(x,y))) .
\end{eqnarray*}
This axiom is sufficient to ensure that the theory of arrays with
$\diff$ admits quantifier-free interpolants for quantifier-free
formulae or, equivalently, that the quantifier-free fragment of the
theory is closed under interpolation.  For example, a
quantifier-free interpolant for $A$ and $B$ above is
\begin{eqnarray*}
  I \coincide x\uguale wr(y,\diff(x,y),rd(x,\diff(x,y))).
\end{eqnarray*}
Notice how $\diff$ permits to represent indexes in the quantifier-free
interpolant $I$ by mentioning only the array constants $a$ and $b$
that are common to $A$ and $B$.  As we will see in the rest of the
paper, this is crucial to compute quantifier-free interpolants.  One
may wonder how useful it is to be able to compute quantifier-free
interpolants in the Skolemized variant of the theory of arrays with
extensionality considered here.  The answer lies in the observation
that this variant is sufficient whenever there is a need to check the
unsatisfiability of formulae as it is the case of many applications;
one of the most important is in model checking procedures for infinite
state systems (see, e.g.,~\cite{jhala}).

\subsection{Contributions}
The paper presents two main contributions,
that are strictly related but completely independent.

First, {we prove non-constructively that given two quan\-ti\-fier-free
  formulae in the the theory of arrays with $\diff$, it is possible to
  compute a quantifier-free interpolant}.  We do this by using the
notion of \emph{amalgamation}~\cite{Hodges,CK}.  Intuitively, a
first-order theory has the amalgamation property if any two structures
in its class of models sharing a common sub-model can be regarded as
sub-structures of a larger model.  A well-known result (see,
e.g.,~\cite{amalgam}) states that if the class of models of a
universal theory $T$ (namely, a theory axiomatized by sentences
obtained by prefixing a quantifier-free formula with a block of
universal quantifiers) have the amalgamation property, then $T$ admits
quantifier-free interpolants for quantifier-free formulae in the
theory and \emph{vice versa}.  Since the theory of arrays with $\diff$
is universal, we consider the problem of showing that its class of
models has the amalgamation property.  We provide a first,
non-constructive, proof of this result by using model-theoretic
notions only.

The {second contribution} of the paper is an {algorithm for the
  generation of quantifier-free interpolants from finite sets}
(intended conjunctively) {of literals in the theory of arrays with
  $\diff$}.  Our algorithm uses as a sub-module a satisfiability
procedure for sets of literals of the theory.  Such a module is based
on a sequence of syntactic manipulations organized in groups of
syntactic transformations.  The most important group of
transformations is a Knuth-Bendix completion procedure (see,
e.g.,~\cite{rewriting}) extended in such a way to solve an equation
$a\uguale wr(b,i,e)$ for $b$ when this is required by the ordering
defined on terms.  (We call Gaussian completion this extended
procedure because of its similarity with the techniques to handle
Gaussian theories~\cite{gaussBG}.)  The goal of these transformations
is to produce what we call a ``\nform'' constraint for which it is
trivial to establish satisfiability.  Given two sets $A$ and $B$ of
literals, the satisfiability procedure is invoked on $A$ and $B$.
While running, the two instances of the procedure exchange literals on
the common signature of $A$ and $B$ (similarly to the Nelson and Oppen
combination method, see, e.g.,~\cite{RRT04}) and perform some
additional actions.  At the end of the computation, the execution
trace is examined and the desired interpolant is built by simple rules
whose goal is to produce a set of literals on the common signature of
$A$ and $B$.  In fact, the problem during the execution of Gaussian
completion is 
to avoid the generation of equalities
containing terms built out of non-shared symbols.  Notice that our
approach seems to be quite different from the standard method of
extracting interpolants from an unsatisfiability proof of $A$ and $B$
in a given calculus (e.g.,~\cite{McM04b,BKR+10}).  
Theoretically, it is not
difficult to refine our proof of termination to show that the proposed
algorithm is in NP, which is optimal since the satisfiability problem
of quantifier-free formulae in the theory of arrays with
extensionality is NP-complete (see, e.g.,~\cite{bradley}).

\subsection{Plan of the paper}

In Section~\ref{sec:background}, we recall some background notions
about theories, model-theoretic notions, and rewriting.  In
Section~\ref{sec:tharr}, we define the theory of arrays with $\diff$,
characterize its models, and show non-constructively that it admits
quantifier-free interpolation.  The rest of the paper is devoted to
prove the same result constructively.  In
Section~\ref{sec:interpolation}, we introduce \nform\ constraints
(which will be manipulated by the interpolation procedure) and state
(and prove) their key properties. In Section~\ref{sec:solver}, we
describe the satisfiability solver for the theory of arrays with
\diff\ based on syntactic transformations of \nform\ constraints.
Then, in Section~\ref{sec:isolver}, we extend such as solver to
produce quantifier-free interpolants by using a carefully designed set
of meta-rules for interpolation.  Finally, in
Section~\ref{sec:related}, we extensively discuss the related work and
conclude.

The appendix contains a proof of the result in~\cite{amalgam} to make
the paper self-contained.

\section{Formal preliminaries}
\label{sec:background}

We assume the usual syntactic (e.g., signature, variable, term, atom,
literal, formula, and sentence) and semantic (e.g.,
structure, 
truth, satisfiability, and validity) notions of first-order logic.
The equality symbol ``\uguale'' is included in all signatures
considered below.  For clarity, we shall use ``\coincide'' in the
meta-theory to express the syntactic identity between two symbols or
two strings of symbols, or to introduce a new definition.
%
%

\subsection{Theories, constraints, interpolants}

A \emph{theory} $T$ is a pair $({\Sigma}, Ax_T)$, where $\Sigma$ is a
signature and $Ax_T$ is a set of $\Sigma$-sentences, called the axioms
of $T$ (we shall 
sometimes write directly $T$ for $Ax_T$).  The $\Sigma$-structures in
which all sentences from $Ax_T$ are true are the \emph{models} of $T$.
A \emph{universal} (resp. \emph{existential}) sentence is obtained by
prefixing a string of universal (resp. existential) quantifiers to a
quantifier-free formula. A theory $T$ is \emph{universal} iff $Ax_T$
consists of universal sentences.
A $\Sigma$-formula $\phi$ is \emph{$T$-satisfiable} if there exists a
model $\cM$ of $T$ such that $\phi$ is true in $\cM$ under a suitable
assignment $\mathtt a$ to the free variables of $\phi$ (in symbols,
$(\cM, \mathtt a) \models \phi$); it is \emph{$T$-valid} (in symbols,
$T\vdash \varphi$) if its negation is $T$-unsatisfiable or,
equivalently, iff $\varphi$ is provable from the axioms of $T$ in a
complete calculus for first-order logic.  A formula $\varphi_1$
\emph{$T$-entails} a formula $\varphi_2$ if $\varphi_1 \to \varphi_2$
is \emph{$T$-valid}; the notation used for such $T$-entailment is
$\varphi_1\vdash_T \varphi_2$ or simply $\varphi_1\vdash \varphi_2$, if $T$ is clear from the context.
The \emph{satisfiability modulo the theory $T$} ($SMT(T)$)
\emph{problem} amounts to establishing the $T$-satisfiability of
quantifier-free $\Sigma$-formulae.


Let $T$ be a theory in a signature $\Sigma$;
a $T$-\emph{constraint} (or, simply, a constraint) $A$ is a set of ground literals in a signature $\Sigma'$
obtained from $\Sigma$ by adding a set of free constants.
Taking conjunction,  we can 
consider a finite constraint $A$ as a single formula; thus, when we
say that a constraint $A$ is \emph{$T$-satisfiable} (or just
``satisfiable'' if $T$ is clear from the context), we mean that the
associated formula
(also called $A$) is satisfiable in a $\Sigma'$-structure which is a
model of $T$. 
Let $a_1, \dots, a_n$ be the tuple of free constants occurring in a
sentence $A$ and $x_1, \dots, x_n$ be a tuple of fresh distinct
individual variables, the formula $A^{\exists}$ is obtained from $A$
by replacing each $a_i$ with $x_i$ (for $i=1, ..., n$) and then
existentially quantifying $x_1, \dots, x_n$, i.e.\ $A^{\exists}$
denotes the formula $\exists x_1\cdots \exists x_n A(x_1/a_1, \dots,
x_n/a_n)$.
We have two notions of equivalence between constraints, which are summarized in the next definition.
\begin{definition}\label{def:constraints} Let $A$ and $B$ be finite constraints (or, more generally, first order sentences) 
in an expanded signature.
We say that 
$A$ and $B$ are \emph{logically equivalent} (modulo $T$) iff $T\vdash A\leftrightarrow B$; on the other hand, we say that they are \emph{$\exists$-equivalent}
(modulo $T$)  iff
$T\vdash A^\exists\leftrightarrow B^\exists$.
\end{definition}

Logical equivalence means that the constraints have the same semantic
content (modulo $T$); $\exists$-equivalence is also useful because we
are mainly interested in $T$-satisfiability of constraints and it is
trivial to see that $\exists$-equivalence implies equisatisfiability
(again, modulo $T$). As an example, if we take a constraint $A$, we
replace all occurrences of a certain term $t$ in it by a fresh
constant $a$ and add the equality $a\uguale t$, 
called the \emph{(explicit) definition (of $t$)}, the constraint $A'$
we obtain in this way is $\exists$-equivalent to $A$.
As another example, suppose that $A\vdash_T a\uguale t$, that $a$ does not occur in $t$, and that $A'$ is 
obtained from $A$ by replacing $a$ by $t$ everywhere; then the following four constraints are
$\exists$-equivalent 
$$
A, \quad A\cup\{ a\uguale t\}, \quad A'	\cup\{a \uguale t\}, \quad A'
$$ (the first three are also pairwise logically equivalent).  The
above examples show how explicit definitions can be introduced and
removed from constraints while preserving $\exists$-equivalence.

A theory $T$ is said to \emph{admit quantifier-free interpolation}
(or, equivalently, to \emph{have quantifier-free interpolants}) iff
for every pair of quantifier free formulae $\phi, \psi$ such that
$\psi\wedge \phi$ is not $T$ satisfiable, there exists a quantifier
free formula $\theta$, 
called an \emph{interpolant}, 
such that: (i) $\psi$ $T$-entails $\theta$; (ii) $\theta \wedge \phi$
is not $T$-satisfiable: (iii) only variables occurring both in $\psi$
and in $\phi$ occur in $\theta$.

\subsection{Some model theoretic concepts and results}
\label{subsec:prelim-model-theory}

We recall some basic model-theoretic notions that will be used in the
paper (for more details, the interested reader is pointed to standard
textbooks in model theory, such as~\cite{CK}).

If $\Sigma$ is a signature, we use the notation $\cM=(M, \mathcal I)$
for a $\Sigma$-structure, meaning that $M$ is the support of $\cM$ and
$\mathcal I$ is the related interpretation function for
$\Sigma$-symbols (in a many-sorted framework, the support is the
disjoint union of the interpretations of the sorts symbols of
$\Sigma$).

Roughly, an embedding is a homomorphism that preserves and reflects
relations and operations.  Formally, a {\it $\Sigma$-embedding} (or,
simply, an embedding) between two $\Sigma$-structu\-res $\cM=(M,
\mathcal I)$ and $\cN=(N,\mathcal J)$ is any mapping $\mu: M
\longrightarrow N$ among the corresponding support sets satisfying the
following three conditions: (a) $\mu$ is a
(sort-preserving) injective function; (b) $\mu$ is an algebraic
homomorphism, that is for every $n$-ary function symbol $f$ and for
every $a_1, \dots, a_n\in M$, we have $f^{\cN}(\mu(a_1), \dots,
\mu(a_n))= \mu(f^{\cM}(a_1, \dots, a_n))$; (c) $\mu$ preserve and
reflects interpreted predicates, i.e.\ for every $n$-ary predicate
symbol $P$, we have $(a_1, \dots, a_n)\in P^{\cM}$ iff $(\mu(a_1),
\dots, \mu(a_n))\in P^{\cN}$. 
By using simple set-theory, it possible to show that every embedding
can be factored in an isomorphism and an inclusion.
This means that if  $\mu$ is an embedding from $\cM$ to
$\cN$, it is possible to assume that 
---up to an isomorphism---$\cM$ is a substructure of $\cN$, in the
sense defined below.

If $M\subseteq N$ and the embedding $\mu: \cM \longrightarrow \cN$ is
just the identity inclusion $M\subseteq N$, we say that $\cM$ is a
{\it substructure} of $\cN$ or that $\cN$ is an {\it superstructure}
of $\cM$. Notice that a substructure of $\cN$ is nothing but a subset
of the carrier set of $\cN$ which is closed under the
$\Sigma$-operations and whose $\Sigma$-structure is inherited from
$\cN$ by restriction. In fact, given $\cN=(N, \mathcal J)$ and
$G\subseteq N$, there exists the smallest substructure of $\cN$
containing $G$ in its carrier set. This is called the substructure
\emph{generated by $G$} and its carrier set can be characterized as
the set of the elements $b\in N$ such that $t^{\mathcal N}(\ua)=b$ for
some $\Sigma$-term $t$ and some finite tuple $\ua$ from $G$ (when we
write $t^{\mathcal N}(\ua)=b$, we mean that $(\cN, \mathtt{a}) \models
t(\ux)\uguale y$ for an assignment $\mathtt{a}$ mapping the $\ua$ to
the $\ux$ and $b$ to $y$).
An easy---but fundamental---fact is that the truth of a universal
(resp. existential) sentence is preserved through substructures
(resp. through superstructures).  

Let $\cM=(M, \mathcal I)$ be a $\Sigma$-structure which is generated
by $G\subseteq M$.  Let us expand $\Sigma$ with a set of fresh free
constants in such a way that in the expanded signature $\Sigma_G$
there is a fresh free constant $c_g$ for every $g\in G$ 
(write $c_g$ directly with $g$ for simplicity).  Let $\cM^G$ be the
$\Sigma_G$-structure obtained from $\cM$ by interpreting each $c_g$ as
$g$.  The \emph{$\Sigma_G$-diagram $\delta_{\cM}(G)$ of $\cM$} is the
set of all ground $\Sigma_G$-literals $L$ such $\cM^G\models L$. When
we speak of the diagram of $\cM$ \emph{tout court}, we mean the
$\Sigma_M$-diagram $\delta_{\cM}(M)$.

The following celebrated result~\cite{CK} is simple, but nevertheless
very powerful and it will be used in the rest of the paper.
\begin{lemma}[Robinson Diagram Lemma]
  \label{lem:robinson}
  Let $\cM=(M, \mathcal I)$ be a $\Sigma$-structu\-re which is generated by
  $G\subseteq M$ and $\cN=(N, \mathcal{J})$ be another
  $\Sigma$-structure.  Then, there is a bijective correspondence between
  $\Sigma$-embeddings
  $\mu:\cM\longrightarrow \cN$ and $\Sigma_G$-expansions $\cN^{(G)}=(N, {\mathcal J}^{(G)})$ of $\cN$ such that
  $\cN^{(G)}\models \delta_{\cM}(G)$. The correspondence associates with $\mu$ the extension of $\mathcal J$
  to $\Sigma_G$ given by ${\mathcal J}^{(G)}(c_g) \coincide \mu(g)$.
\end{lemma}
Notice that an embedding $\mu:\cM \longrightarrow \cN$ is uniquely
determined, in case it exists, by the image of the set of generators
$G$: this is because the fact that $G$ generates $\cM$ implies (and is
equivalent to) the fact that every $c\in M$ is of the kind
$t^\cM(\ug)$, for some term $t$ and some $\ug$ from $G$.

Intuitively, amalgamation is a property of collections of structures
that guarantees
that two structures in the collection can be glued
into substructures of a larger one.
Formally, a theory $T$ is said to have the \emph{amalgamation
  property} iff whenever we are given embeddings
\begin{eqnarray*}
\mu_1:\cN \longrightarrow \cM_1, \qquad \mu_2:\cN \longrightarrow \cM_2
\end{eqnarray*}
among the models $\cN, \cM_1, \cM_2$ of $T$, then there exists a
further model $\cM$ of $T$ endowed with embeddings
\begin{eqnarray*}
\nu_1:\cM_1 \longrightarrow \cM, \qquad \nu_2:\cM_2 \longrightarrow \cM
\end{eqnarray*}
such that $\nu_1\circ \mu_1= \nu_2\circ \mu_2$. Notice that, up to
isomorphism, we can limit ourselves in the above definition to the
case in which $\mu_1, \mu_2$ are inclusions, i.e. to the case in which
$\cN$ is just a substructure of both $\cM_1, \cM_2$; in this case,
$\cM$ is said to be a $T$-amalgam of $\cM_1$ and $\cM_2$ over
$\cN$. ({When the signature does not have ground terms of some sort,
  models $\cN$ having empty domain(s) must be included in the
  definition of amalgamation property.)
\begin{theorem}[\cite{amalgam}]
  \label{thm:interpolation-amalgamation} 
  Let $T$ be universal; then $T$ admits quantifier free interpolants
  iff $T$ has the amalgamation property.
\end{theorem}
We emphasize that the hypothesis for $T$ to be universal is necessary
for the above result to hold.  To make the paper self-contained, we
include the proof of this result in Appendix~\ref{app:appknown}.

\subsection{Some term rewriting concepts and results}

We shall need basic term rewriting system notions and results (see,
e.g.,~\cite{rewriting}).  In the following, we recall some of the most
important ones for this paper.  

The reflexive and transitive closure of a binary relation
$\rightarrow$ is denoted with $\rightarrow^*$ and its transitive
closure by $\rightarrow^+$.  A binary relation $\rightarrow$ over a
set $E$ is \emph{terminating} if there are no infinite sequence $e_0,
e_1, ...$ of elements of $E$ such that $(e_i,e_{i+1})\in \rightarrow$,
also written as $e_i\rightarrow e_{i+1}$, for every $i\geq 0$.  The
relation $\rightarrow\subseteq E\times E$ is \emph{confluent} if there
exists $v\in E$ such that $s\rightarrow^* v$ and $t\rightarrow^* v$
whenever $u\rightarrow^* s$ and $u\rightarrow^* t$, for $s,t,u\in E$.
The relation $\rightarrow$ is \emph{convergent} if it is both
terminating and confluent.


A \emph{rewrite rule} is an ordered pair 
of terms $l$ and
$r$, written as $l\rightarrow r$ (intuitively, the rule is used to
replace instances of $l$ with instances of $r$).\footnote{
To avoid pathological cases, it is assumed that all variables
occurring in $r$ occur also in $l$.  }  A \emph{(term-)rewriting
  system}
is a set $R$ of rewrite rules, 
which
induces a \emph{rewrite relation} $\rightarrow_R$ (or simply
$\rightarrow$ when $R$ is clear from the context) on terms as follows:
$\rightarrow_R$ is the 
relation 
that contains the pairs of terms $(t, t')$ such that (for some
$l\rightarrow r$  in $R$)
the term $t$ 
%
has a
sub-term 
of the form
$l\sigma$ for some substitution $\sigma$ (in symbols
$t\equiv t[l\sigma]$), 
and $t'$ is obtained by replacing
that subterm
$l\sigma$ 
by $r\sigma$ in $t$
(in symbols $t'\equiv t[r\sigma]$).
Let $s$ and $t$ terms; we say that
$s$ and $t$ are \emph{joinable} w.r.t.\ a rewrite relation
$\rightarrow$ (in symbols, $s\downarrow t$) when there exists a term
$u$ such that $s\rightarrow^* u$ and $t\rightarrow^* u$.  A term $t$
is \emph{reducible} w.r.t.\ a rewrite relation $\rightarrow$ if there
exists a term $u$ such that $t\rightarrow u$; otherwise, $t$ is
\emph{irreducible}.  A term $u$ is a \emph{normal form} of $t$
w.r.t.\ a rewrite relation $\rightarrow$ if $t\rightarrow^* u$ and $u$
is irreducible.  A rewrite relation is \emph{ground convergent} when
it is convergent
once restricted to
the set of ground terms.  Convergent rewrite
relations are interesting because they have unique normal forms. 
\emph{Knuth–Bendix completion} is a procedure,
based on superposition of critical pairs, for transforming a rewrite
system into a confluent one (see, e.g.,~\cite{rewriting} for details).
Termination of rewrite systems is undecidable.

A \emph{quasi-ordering} is a reflexive and transitive relation.  The
\emph{lexicographic path ordering} $\succ$ on a set of terms induced
by a quasi-ordering $>$, called \emph{precedence relation}, on the set
of constant and function symbols on which the terms are built is
defined as follows: $s = f(s_l, \ldots, s_m) \succ g(t_l, \ldots, t_n)
= t$ iff
\begin{enumerate}[(1)]
  \item $s_k \succ t$ 
   or $s_k\equiv t$
   for some $k\in \{1, \ldots, m\}$, or
  \item $f > g$ and $s \succ t_l$ for each $l\in \{1, \ldots, n\}$, or
  \item $f \equiv g$, $s_1 \equiv t_1$, ..., $s_{j-1} \equiv t_{j-1}$, $s_{j} \succ
    t_{j}$, $s\succ t_{j+1}$, ..., $s\succ t_{n}$ for some $j\in \{1,
    \ldots, n\}$.
\end{enumerate}
If the precedence relation $>$ is also total, then so is $\succ$
once restricted to ground terms.

\section{Theories of Arrays and Quantifier-free Interpolation}
\label{sec:tharr}

The McCarthy \emph{theory of arrays} \AX~\cite{mccarthy} has three
sorts $\ARRAY, \ELEM, \INDEX$ (called ``array'', ``element'', and
``index'' sort, respectively) and two function symbols $rd$ and $wr$
of appropriate arities; its axioms are:
\begin{eqnarray}
  \label{ax1}
  \forall y, i, e. & & rd(wr(y,i,e),i) \uguale e \\
  \label{ax2}
  \forall y, i, j, e. & & i \not\uguale j \imp rd(wr(y,i,e),j)\uguale rd(y,j) .
\end{eqnarray}
The theory of \emph{arrays with extensionality} \AXEXT has the further
axiom
\begin{eqnarray*}
 \forall x, y. 
 x \not\uguale y \imp (\exists i.\ rd(x,i)\not\uguale rd(y,i)),
\end{eqnarray*}
called the `extensionality' axiom.  In this paper, we consider a
variant of the McCarthy {theory of arrays} with extensionality,
obtained by Skolemizing the axioms of extensionality.  Formally, we
define the \emph{theory of arrays with $\diff$} \AXDIFF by adding the
additional (Skolem) function $\diff$ to the signature of \AXDIFF and
replace the extensionality axiom by its Skolemization, namely
\begin{eqnarray}
  \label{ax3a}
  \forall x, y. & & x \not\uguale y \imp rd(x,\diff(x,y))\not\uguale rd(y,\diff(x,y)) .
\end{eqnarray} 
The new symbol $\diff$ is binary and takes two arguments of sort
$\ARRAY$ and returns an element of sort $\INDEX$.  The new axiom
(\ref{ax3a}) constrains $\diff$ to return an index at which the two
arrays in input store different values, 
whereas it
returns an arbitrary
value when input arrays are equal.  

\subsection{A semantic argument for quantifier-free interpolation}
\label{subsec:semantic-arg}

Here, we show that \AXDIFF does admit quantifier-free interpolation,
contrary to \AXEXT~\cite{KMZ06}.  We do so by using a model-theoretic
argument based on the equivalence between amalgamation of the models
and admitting quantifier-free interpolation for universal theories
(recall Theorem~\ref{thm:interpolation-amalgamation} in
Section~\ref{subsec:prelim-model-theory}).  Notice that \AXDIFF is
universal whereas \AXEXT is not.

Since amalgamation is a property of the models of a theory, we
preliminarily discuss the class of models of \AXDIFF. 
A model of \AXEXT or \AXDIFF is \emph{standard} when
$\ARRAY$ is interpreted as the set of all functions from indexes to
elements.  In a standard model of \AXEXT or \AXDIFF, arrays are
interpreted as functions, $rd$ as function application, and $wr$ as
the point-wise update operation (i.e.\ the interpretation of
$wr(a,i,e)$ returns the same values of the interpretation of $a$,
except at the interpretation of index $i$ where it returns the
interpretation of $e$).  Indeed, the class of models of \AXEXT or
\AXDIFF contains also non-standard models.  This is because the
axioms of both \AXEXT and \AXDIFF, being first-order formulae, do not
constrain the interpretation of the sort $\ARRAY$ to contain all
mappings from indexes to elements.  (This is similar to the
interpretation of function variables according to the Henkin semantics
of second order logic; see, e.g.,~\cite{enderton}.)  Fortunately,
because of the extensionality axiom, it is easy to show 
(see below)
that every
model of such theories embeds into a standard one (recall the
definition of embedding in Section~\ref{subsec:prelim-model-theory}).
This means that any model is isomorphic to a sub-structure of a
standard model in which arrays are interpreted as functions, although
it might happen that not all functions are part of the interpretation
of $\ARRAY$ in the model. 
As a consequence, whenever
we want to test the validity of universal formulae or the
satisfiability of constraints, we can---w.l.o.g.---consider only
standard models.  (This fact will be used in the proofs of some
results in later sections, such as the proof of Lemma~\ref{lem:normal}
where a standard model is built to show the satisfiability of a
certain class of constraints of \AXDIFF.) 

We show that the universal theory \AXDIFF has the amalgamation
properties so that, by Theorem~\ref{thm:interpolation-amalgamation},
we are entitled to conclude that it admits quantifier-free
interpolation.  Recall from Section~\ref{subsec:prelim-model-theory}
that a universal theory has the amalgamation property if two of its
models can be
glued as substructures of a third model.  Thus, we
need to consider arbitrary models of \AXDIFF, not only the standard
ones.
This is why we need more insight into arbitrary models of our theories and their relationship to standard ones.

Let us choose an arbitrary model $\cM$ of  \AXEXT. 
We can build the standard model $std(\cM)$ such that $\INDEX^{std(\cM)}=\INDEX^\cM$ and $\ELEM^{std(\cM)}=\ELEM^{\cM}$. To embed $\cM$ into 
$std(\cM)$ is sufficient to associate with every $a\in\ARRAY^{\cM}$  the function
mapping $i$ to $rd^{\cM}(a,i)$ (this is an embedding because of the extensionality axiom). In this way, we can identify $\ARRAY^{\cM}$ with a subset of the set of all functions
$\ARRAY^{std(\cM)}$. If we call \emph{functional} a model $\cM$ in which $\ARRAY^{\cM}$ is a subset of the set of functions from $\INDEX^{\cM}$ to $\ELEM^{\cM}$
(and in which $rd^\cM, wr^\cM$ have the standard meaning), we have just shown that \emph{every model 
is isomorphic to a functional one}.
(The argument extends to models of \AXDIFF although---in a standard
model---the interpretation of \diff is not fixed as the
interpretations of $rd$ and $wr$.)
In this respect, the crucial  question is the following: which subsets of the set $\ARRAY^{\bar{\cM}}$ in a standard model $\bar{\cM}$ can be in the support $\ARRAY^{\cM}$ of a functional model $\cM$
(with $\INDEX^\cM=\INDEX^{\bar{\cM}}$, $\ELEM^\cM=\ELEM^{\bar{\cM}}$)
that is a 
substructure   
of $\bar{\cM}$? 
We shall answer the question by using the notion of ``closure under cardinality dependence,'' that we formally define next.

%
 

Let $a,b$ be
elements of $\ARRAY^\cM$ in a model $\cM$ of \AXDIFF.  We say that
\emph{$a$ and $b$ are cardinality dependent} (in symbols, $\cM\models
|a-b| < \omega$) iff $\{i\in \INDEX^\cM \mid \cM\models rd(a,i)\neq
rd(b,i)\}$ is finite.  
Cardinality dependency is obviously an equivalence relation.
\begin{lemma}
  \label{lem:dependency}
  Let $\cN$, $\cM$ be models of \AXDIFF such that $\cM$ is a
  substructure of $\cN$.  For every $a,b\in\ARRAY^\cM$, we have that
  \begin{eqnarray*}
    \cM\models |a-b| < \omega & \mbox{ {\rm iff} } &
    \cN\models |a-b| < \omega.
  \end{eqnarray*}
\end{lemma}
\begin{proof}
The right-to-left side is trivial because if $\cM\models |a-b| <
\omega$ then $\cM\models a=wr(b, I, E)$,
where $I\coincide i_1, \ldots, i_n$ is a list of terms of sort
$\INDEX$, $E \coincide e_1, \ldots, e_n$ is a list of terms of sort
$\ELEM$, and $wr(b, I, E)$ abbreviates the term $wr(wr(\cdots wr (a,
i_1, e_1) \cdots), i_n, e_n)$ (this and similar notations will be
discussed in more details in Section~\ref{sec:interpolation}).  Thus,
also $\cN\models a=wr(b, I, E)$ because $\cM$ is a substructure of
$\cN$. Vice versa, suppose that $\cM\not\models |a-b| < \omega$. This
means that there are infinitely many $i\in\INDEX^\cM$ such that
$rd^\cM(a,i)\neq rd^\cM(b,i)$. Since $\cM$ is a substructure of $\cN$,
there are also infinitely many $i\in\INDEX^\cN$ such that
$rd^\cN(a,i)\neq rd^\cN(b,i)$, i.e.  $\cN\not\models |a-b| < \omega$.
\end{proof}


We are now in the position to show how any functional model $\cM$ of
\AXDIFF (i.e. up to isomorphism, any model whatsoever) can be obtained
from a standard one.
In order to produce any such $\cM$, it is sufficient to take a
standard model $\bar{\cM}$, to let $\INDEX^\cM \coincide \INDEX^{\bar{\cM}}$, $\ELEM^\cM \coincide
\ELEM^{\bar{\cM}}$, and to let $\ARRAY^\cM$ to be equal to any subset of
$\ARRAY^{\bar{\cM}}$ that is \emph{closed under cardinality dependence},
i.e.\ such that if $a\in \ARRAY^\cM$ and $\bar{\cM}\models |a-b| < \omega$,
then $b$ is also in $\ARRAY^\cM$.  In other words, functional
substructures $\cM$ of $\bar{\cM}$ with $\INDEX^\cM=\INDEX^{\bar{\cM}}$ and
$\ELEM^\cM= \ELEM^{\bar{\cM}}$ are in bijective correspondence with subsets of
$\ARRAY^{\bar{\cM}}$ closed under cardinality dependence.
%

 A similar remark holds for embeddings.
Suppose that $\mu: \cN\longrightarrow \cM$ is an embedding that
restricts to an inclusion $\INDEX^\cN\subseteq \INDEX^\cM$,
$\ELEM^\cN\subseteq \ELEM^\cM$ for $\cM$ and $\cN$ 
functional 
models of \AXDIFF.
The action of the embedding $\mu$ on $\ARRAY^\cN$ can be characterized
as follows: take an element $a$ for each cardinality dependence
equivalence class, extend arbitrarily $a$ to the set
$\INDEX^\cM\setminus\INDEX^\cN$ to produce $\mu(a)$ and then define
$\mu(b)$ for non representative $b$ in the only possible way for $wr$
to be preserved; i.e.\ if $\cN\models b=wr(a, I, E)$ for a
representative $a$, let $\mu(b)$ be $wr^{\cM}(\mu(a), I, E)$.

By using the observation above, we are ready to show that $\AXDIFF$
has the amalgamation property.
\begin{theorem}
  \label{thm:amalgamat}
  The theory $\AXDIFF$ has the amalgamation property.
\end{theorem}
\begin{proof}
  Take two embeddings $\mu_0:\cN\longrightarrow \cM_0$ and
  $\mu_1:\cN\longrightarrow \cM_1$.  As observed above, we can
  suppose---w.l.o.g.---that $\cN, \cM_0, \cM_1$ are 
  functional models,
  that $\mu_0, \mu_1$ restricts to inclusions for the sorts $\INDEX$
  and $\ELEM$, and that $(\ELEM^{\cM_0}\setminus \ELEM^\cN) \cap
  (\ELEM^{\cM_1}\setminus \ELEM^\cN)=\emptyset$,
  $(\INDEX^{\cM_0}\setminus \INDEX^\cN) \cap (\INDEX^{\cM_1}\setminus
  \INDEX^\cN)=\emptyset$.  To simplify our task, we can also
  suppose---again w.l.o.g.---that there exists some
  $e_i\in(\ELEM^{\cM_i}\setminus \ELEM^\cN)$ and some $j_i\in
  (\INDEX^{\cM_i}\setminus \INDEX^\cN)$ (i.e. that these sets are not
  empty), for $i=0,1$. ({If this additional condition is not
    satisfied, it is sufficient to enlarge $\cM_1, \cM_2$ so that they
    satisfy it.})  The amalgamated model $\cM$ will be the standard
  model over $\INDEX^{\cM_0}\cup \INDEX^{\cM_1}$ and
  $\ELEM^{\cM_0}\cup \ELEM^{\cM_1}$.  We need to define
  $\nu_i:\cM_i\longrightarrow \cM$ ($i=0,1$) in such a way that
  $\nu_0\circ \mu_0=\nu_1 \circ \mu_1$.  The only relevant point is
  the action of $\nu_i$ on $\ARRAY^{\cM_i}$: as observed above, in
  order to define it, it is sufficient to extend any $a\in
  \ARRAY^{\cM_i}$ to the indexes $k\in(\INDEX^{\cM_{1-i}}\setminus
  \INDEX^\cN)$:
  \begin{enumerate}[(I)]
  \item we let the value $\nu_i(a)(k)$ be $e_{i}$ in case
    there is no $c$ such that $\cM_i\models | a -\mu_i(c)| < \omega$;
  \item otherwise, we can do the following: take any such
    $c$ such that $\cM_i\models | a -\mu_i(c)| < \omega$
    and put $\nu_i(a)(k) \coincide \mu_{1-i}(c)(k)$.
  \end{enumerate} 
  Because of Lemma~\ref{lem:dependency} the choice of $c$ in (II)
  above is immaterial.  In fact, any other $c'$ differs from $c$ only
  w.r.t.\ a finite set of indices in $\cM_i$.  This also holds in
  $\cN$ (by Lemma~\ref{lem:dependency}) and thus we have $\cN\models
  c'=wr(c, I, E)$ for some $I\subseteq \INDEX^\cN$. The latter implies
  that $\mu_{1-i}(c)$ and $\mu_{1-i}(c')$ cannot differ at any
  $k\in(\ELEM^{\cM_{1-i}}\setminus \ELEM^\cN)$.  This guarantees that
  $\nu_1\circ \mu_1=\nu_2 \circ \mu_2$.

  In order to define $\diff^\cM$ we can simply extend
  $\diff^{\cM_1}\cup \diff^{\cM_2}$ in such a way that
  axiom~\ref{ax3a} holds. More precisely, we define $\diff^\cM(a,b)$
  as follows: (i) if for some $i=0,1$, we have that $a=\nu_i(a')$ and
  $b=\nu_i(b')$, then $\diff^\cM(a,b)$ is taken to be
  $\diff^{\cM_i}(a', b')$; (ii) otherwise it is defined to be any $i$
  such that $a(i)\neq b(i)$ 
 (it is arbitrary whenever
  $a=b$). For this definition of $\diff^\cM$ to be correct, 
  it is sufficient 
  to show that
  \begin{desCription}
  \item\noindent{\hskip-12 pt\bf Claim:}\ \emph{if $a=\nu_0(a_0)=\nu_1(a_1)$, then there exists
    $c$ such that $a_0=\mu_0(c)$ and $a_1=\mu_1(c)$}.
  \end{desCription}
  To prove the claim, suppose that $a=\nu_0(a_0)=\nu_1(a_1)$. Then
  $\nu_0(a_0)$ and $\nu_1(a_1)$ must have been defined as in (II)
  above (otherwise they cannot coincide with each other at indexes
  $j_0, j_1$),\footnote{
  The Claim might be false in case $\INDEX^{\cM_1}=\INDEX^{\cN}=\INDEX^{\cM_2}$, this is the reason why we enlarged 
  $\INDEX^{\cM_1},\INDEX^{\cM_2}$
  by adding the extra indexes $j_0,j_1$.
  } which means that there exists $c_i$ such that for
  $i=0,1$ we have $\cM_i\models | a_i -\mu_i(c_i)| < \omega$. Since
  $\nu_0(a_0)=a=\nu_1(a_1)$, this means that
  $\nu_0(\mu_0(c_0))=\nu_1(\mu_1(c_0))$ and $a$ differ only at
  finitely many indexes; the same is true for $\nu_1(\mu_1(c_1))$ and
  $a$, which in turns implies that $\nu_1(\mu_1(c_0))$ and
  $\nu_1(\mu_1(c_1))$ differ only at finitely many indexes too. The
  same consequently holds for $c_0, c_1$ in $\cN$ too, for
  $\mu_0(c_0)$ and $\mu_0(c_1)$ in $\cM_0$ and for $\mu_1(c_0)$ and
  $\mu_1(c_1)$ in $\cM_1$.  Thus, since the choice of $c$ in (II) is
  immaterial, we can suppose---w.l.o.g.---that 
  $c_0\uguale c_1$ (let us use just $c$ to name it).
  Then, by
  (II) applied to the definition of $\nu_1(a_1)$, we have that
  $\nu_0(\mu_0(c))=\nu_1(\mu_1(c))$ and $a=\nu_1(a_1)$ cannot differ
  at any $k\in (\ELEM^{\cM_{0}}\setminus \ELEM^\cN)$. Similarly,
  $\nu_0(\mu_0(c))=\nu_1(\mu_1(c))$ and $a$ cannot differ at any $k\in
  (\ELEM^{\cM_{1}}\setminus \ELEM^\cN)$. Thus $a$ and
  $\nu_0(\mu_0(c))=\nu_1(\mu_1(c))$ possibly differ only for $k\in
  \INDEX^\cN$ and actually only for finitely many such $k$. But
  $a=\nu_0(a_0)=\nu_1(a_1)$, so the values of $a$ at any $k\in
  \INDEX^\cN$ belongs $\ELEM^{\cM_0}\cap \ELEM^{\cM_1}=\ELEM^\cN$,
  which means that $a$ is equal to
  $wr^\cM(\nu_0(\mu_0(c)),I,E)=\nu_0(\mu_0(wr^\cN(c, I, E)))$ for
  $I\subseteq \INDEX^\cN$ and $E\subseteq \ELEM^\cN$. In conclusion,
  we have that $a$ is of the kind $\nu_0(\mu_0(\tilde
  c))=\nu_1(\mu_1(\tilde c))$ and from $a=\nu_0(a_0)=\nu_1(a_1)$, we
  get $a_0=\mu_0(\tilde c)$ and $a_1=\mu_1(\tilde c)$ because $\nu_0,
  \nu_1$ are injective.  
\end{proof}
Before stating the main result of the paper which immediately follows
from Theorems~\ref{thm:interpolation-amalgamation}
and~\ref{thm:amalgamat}, it is interesting to observe the following
about the Claim used in the proof of Theorem~\ref{thm:amalgamat}.
The property mentioned in the Claim is known as \emph{strong
  amalgamability property} in Universal Algebra and is key to derive
quantifier-free interpolation in combination of
theories~\cite{strong_amalgamability}.  The fact that \AXDIFF enjoys
strong amalgamability is crucial to transfer quantifier-free
interpolation to combinations of \AXDIFF with other important
theories, like equality with uninterpreted symbols, difference logic,
real arithmetic, appropriate variants of integer linear arithmetic,
etc.  We refer the reader to~\cite{strong_amalgamability} for details.
\begin{theorem}
  \label{thm:amalgamation}
  The theory $\AXDIFF$ admits quantifier-free interpolation.
\end{theorem}
%
We conclude this section with some observations concerning the 
theories \AXEXT and \AX.
Lemma~\ref{lem:dependency} holds also for the theory
\AXEXT and the proof of Theorem~\ref{thm:amalgamat} goes through also for
\AXEXT.  However,
according to Theorem~\ref{thm:interpolation-amalgamation} in
Section~\ref{subsec:prelim-model-theory}, amalgamation alone is not
sufficient for establishing quantifier-free interpolation for theories
like \AXEXT which are not universal 
(for non universal theories one needs sub-amalgamability, not just amalgamability, see~\cite{strong_amalgamability}).
  Indeed, \AXEXT is amalgamable
but does not admit quantifier-free interpolation.

Despite being universal, \AX is not amalgamable and thus it does not
admit quantifier-free interpolation. Indeed, the left-to-right
implication of Lemma~\ref{lem:dependency} does not hold for \AX as 
the arguments in the proof of Theorem~\ref{thm:amalgamat}. 
To get a formal counterexample to the amalgamability of \AX,
consider the following situation.  
Let $\cN$ be the \AX-model in which $\ELEM^\cN$ and $\INDEX^\cN$ are
empty and $\ARRAY^\cM$ contains two distinct elements, say $a$ and
$b$.  As already observed, empty supports must be taken into account
when showing the amalgamation property and, for \AX, the axiom of
extensionality needs not be satisfied.  Extend $\cN$ to two standard
models $\cM_1$ and $\cM_2$, where $\ELEM^{\cM_1}=\{e, e'\},
\INDEX^{\cM_1}=\{ i\}$ and $\ELEM^{\cM_2}=\{d_1,d_2\},
\INDEX^{\cM_2}=\{ j_1,j_2\}$.  Then, embed $\cN$ into $\cM_1$ by
letting $a,b$ differ at $i$ (thus, e.g., $\cM_1\models a=wr(b, i, e)
\wedge rd(b,i)=e'$) and embed $\cN$ into $\cM_2$ by letting $a,b$
differ at both $j_1$ and $j_2$. Now, observe that amalgamation fails
because we should have
\begin{eqnarray*}
  \cM\models a=wr(b, i, e)\wedge rd(a,j_1)\not= rd(b, j_1)\wedge
  rd(a,j_2)\not= rd(b, j_2)\wedge j_1\not= j_2
\end{eqnarray*}
in any amalgamated model $\cM$ and this is in contradiction with the
two axioms of \AX.
%

\section{Modular constraints for Arrays with \diff\ and their combinations}
\label{sec:interpolation}

Theorem~\ref{thm:amalgamation} is proved by semantic arguments, hence
it does not give an interpolation algorithm; it only guarantees that,
by enumerating quantifier free formulae, one can find sooner or later
the desired interpolant.  In the rest of the paper, we develop
(\emph{independently} of the results of Section~\ref{sec:tharr})
techniques based on rewriting and constraint solving to construct an
algorithm computing quantifier-free interpolants for conjunctions of
ground literals in $\AXDIFF$.  Here, we introduce the notion of
``\nform constraint,'' which is the main data structure manipulated by
the quantifier-free interpolation procedure and we prove two key
properties.  First, we show that the satisfiability of \nform
constraints can be easily detected (Lemma~\ref{lem:normal}).  Second,
we prove that they can be combined in a modular way
(Proposition~\ref{prop:merging}).


Preliminarily, we introduce some notational conventions which are
specific for constraints in the theory \AXDIFF.  We use $a, b, \dots$
to denote free constants of sort \ARRAY, $i, j, \dots$ for free
constants of sort \INDEX, and $d, e, \dots$ for free constants of sort
\ELEM; $\alpha, \beta, \dots$ stand for free constants of any sort.
Below, we shall introduce non-ground rewriting rules involving
(universally quantified) variables of sort \ARRAY: for these
variables, we shall use the symbols $x, y, z, \dots$.  We make use of
the following abbreviations.
\begin{iteMize}{$-$}
 \item[-] [Nested write terms] By $wr( a, I, E )$ we indicate a nested
   write on the array variable $a$, where indexes are represented by
   the free constants list $I \coincide i_1, \ldots, i_n$ and elements
   by the free constants list $E \coincide e_1, \ldots, e_n$; more
   precisely, $wr( a, I, E )$ abbreviates the term $wr(wr(\cdots wr
   (a, i_1, e_1) \cdots), i_n, e_n)$.  Notice that, whenever the
   notation $wr( a, I, E )$ is used, the lists $I$ and $E$ must have
   the same length; for empty $I, E$, the term $wr(a,I,E)$
   conventionally stands for $a$.
 \item[-] [Multiple read literals] Let $a$ be a constant of sort \ARRAY,
   $I \coincide i_1, \ldots, i_n$ and $E \coincide e_1, \ldots, e_n$
   be lists of free constants of sort \INDEX\ and \ELEM, respectively;
   $rd(a,I)\uguale E$ abbreviates the formula $rd(a,i_1)\uguale
   e_1\wedge \cdots\wedge rd(a,i_n)\uguale e_n$.
 \item[-] [Multiple equalities] If $L\coincide\alpha_1, \dots, \alpha_n$
   and $L'\coincide\alpha'_1, \dots, \alpha'_n$ are lists of constants
   of the same sort, by $L=L'$ we indicate the formula
   $\bigwedge_{i=1}^n \alpha_i\uguale \alpha'_i$.
 \item[-] [Multiple distinctions] If $L\coincide\alpha_1, \dots,
   \alpha_n$ is a list of constants of the same sort, by $Distinct(L)$
   we abbreviate the formula $\bigwedge_{i\neq j} \alpha_i \not\uguale
   \alpha_j$.
 \item[-] [Juxtaposition and subtraction] If $L\coincide\alpha_1, \dots,
   \alpha_n$ and $L'\coincide\alpha'_1, \dots, \alpha'_m$ are lists of
   constants, by $L\cdot L'$ we indicate the list $\alpha_1, \dots,
   \alpha_n, \alpha'_1, \dots, \alpha'_m$; for $1\leq k\leq n$, the
   list $L-k$ is the list $\alpha_1, \dots,\alpha_{k-1}, \alpha_{k+1},
   \dots, \alpha_n$.
\end{iteMize}
\begin{figure}[t]
    \centering
  \begin{eqnarray*}
    \begin{array}{|r|l|}
      \hline 
    \mbox{\textbf{Refl} } &
    wr(a, I, E)\uguale a
     \leftrightarrow  
    rd(a,I)\uguale E 
    \\
    & 
      \mbox{\emph{Proviso}: $Distinct(I)$}
    \\ \hline
    \mbox{\textbf{Symm} } &
    (wr(a, I, E)\uguale b \wedge rd(a,I)\uguale D) 
     \leftrightarrow 
    (wr(b, I, D)\uguale a \wedge rd(b,I)\uguale E) 
    \\
    & 
    \mbox{\emph{Proviso}:  $Distinct(I)$}
     \\ \hline
    \mbox{\textbf{Trans} } &
    (a\uguale wr(b, I, E) \wedge b\uguale wr(c, J, D))
     \leftrightarrow  
    (a\uguale wr(c, J\cdot I, D\cdot E) \wedge b\uguale wr(c, J, D)) 
    \\  \hline
    \mbox{\textbf{Confl} } &
    b\uguale wr(a, I\cdot J, E\cdot D) \wedge b\uguale wr(a, I\cdot H, E'\cdot F) \leftrightarrow \\
    & \leftrightarrow (b\uguale wr(a, I, E)  \wedge
    E\uguale E' \wedge 
    rd(a,J)\uguale D \wedge 
    rd(a,H)\uguale F 
    ) \\
    & 
    \mbox{\emph{Proviso}:  $Distinct(I\cdot J\cdot H)$}
     \\ \hline
    \mbox{\textbf{Red} } &
    (a\uguale wr(b, I, E) \wedge  rd(b,i_k)\uguale e_k)
     \leftrightarrow  
    (a\uguale wr(b, I\!-\!k, E\!-\!k) \wedge  rd(b,i_k)\uguale e_k) \\
    & 
   \mbox{\emph{Proviso}:  $Distinct(I)$}
     \\ \hline
    \end{array}
  \end{eqnarray*}
  \begin{minipage}{.95\textwidth}
    \emph{Legenda}: $a$ and $b$ are constants of sort \ARRAY;
    $I\coincide i_1, \dots, i_n$, $J\coincide j_1, \dots, j_m$ and   $H\coincide h_1, \dots, h_l$ are lists of
    constants of sort \INDEX;  $E\coincide e_1, \dots, e_n$, $E'\coincide e_1',
    \dots, e_n'$, $D\coincide d_1, \dots, d_m $,  and $F\coincide f_1, \dots, f_l$
    are lists of constants of sort \ELEM.
  \end{minipage}
  \caption{\label{fig:key-lemmas}Key properties of write terms}
\end{figure}
Some key properties of equalities involving write terms are stated in
the following lemma (see also Figure~\ref{fig:key-lemmas}).
\begin{lemma}[Key properties of write terms]
  \label{lem:key-lemmas} 
  The formulae in Figure~\ref{fig:key-lemmas} are all \AXDIFF-valid
  under the assumption that their provisoes - if any - hold (when we
  say that a formula $\phi$ is \AXDIFF-valid under the proviso $\pi$,
  we just mean that $\pi\vdash_{\AXDIFF} \phi$).
\end{lemma}
\begin{proof}
  The properties in Figure~\ref{fig:key-lemmas} are all
  straightforward to derive.  Here, we just sketch the proof of
  \textbf{Trans}itivity, as an example: one side is by replacement of
  equals; for the-right-to-left side, notice that the equalities
  $a\uguale wr(c, J\cdot I, D\cdot E)$ and $b\uguale wr(c, J, D)$ can
  be used as rewrite rules to rewrite both members of $a\uguale wr(b,
  I, E)$ to the same term.
\end{proof}

\subsection{Modular constraints in $\AXDIFF$}
A (ground) \emph{flat} literal is a literal of the form $a\uguale wr(
b, I, E ), rd(a, i)\uguale e, \diff(a,b)\uguale i, \alpha\uguale\beta,
\alpha\not\uguale \beta$.  Notice that replacing a sub-term $t$ with a
fresh constant $\alpha$ in a constraint $A$ and adding the
corresponding defining equation $\alpha\uguale t$ to $A$ always
produces an $\exists$-equivalent constraint; by repeatedly applying
this method, one can show that every constraint is
$\exists$-equivalent to a \emph{flat} constraint, i.e., to one
containing only flat literals.  We split a flat constraint $A$ into
two parts, the \emph{index} part $A_I$ and the \emph{main} part $A_M$:
$A_I$ contains the literals of the form $ i\uguale j, i\not\uguale j,
\diff(a,b)\uguale i, $ whereas $A_M$ contains the remaining literals,
i.e., those of the form $ a\uguale wr(b, I, E), a\not\uguale b, rd (a,
i)\uguale e, e\uguale d, e\not\uguale d $ (atoms $a\uguale b$ are
identified with literals $a\uguale wr(b, \emptyset, \emptyset)$).  We
write $A = <A_I, A_M>$ to indicate the two parts of the constraint
$A$.  In the main part of a constraint, positive literals will be
treated as rewrite rules; to get a suitable orientation, we use a
\emph{lexicographic path ordering} with a total precedence $>$ such
that $ a>wr>rd>\diff>i>e, $ for all $a, i, e$ of the corresponding
sorts.  This choice orients equalities $a\uguale wr(b, I, E)$
\emph{from left to right} when $a>b$; equalities like $a\uguale wr(b,
I, E)$ for $a< b$ or $a\coincide b$ will be called \emph{badly
  orientable} equalities.  
\begin{definition}
  \label{def:normal}
  A constraint $A= <A_I, A_M>$ is said to be \emph{\nform} iff
  it is flat and the following conditions are satisfied (we let
  $\tilde I, \tilde E$ be the sets of free constants of sort $\INDEX$
  and $\ELEM$ occurring in $A$):
  \begin{desCription}
  \item\noindent{\hskip-12 pt\rm\phantom{i}(o)}\  no positive index literal $i\uguale j$ occurs in $A_I$;
   \item\noindent{\hskip-12 pt\rm\phantom{ii}(i)}\  no negative array literal $a\not\uguale b$ occurs
     in $A_M$;
  \item\noindent{\hskip-12 pt\rm\phantom{i}(ii)}\ $A_M$ does not contain badly orientable equalities;
  \item\noindent{\hskip-12 pt\rm (iii)}\ the rewriting system $A_R$ given by the oriented
    positive literals of $A_M$ joined with the rewriting rules
    \begin{eqnarray}
      rd(wr(x, i, e),j) \rightarrow rd(x,j)  &\hskip .6cm {\hbox{\rm for $i, j\in \tilde I$, $e\in \tilde E$, $i\not\coincide j$}}\label{eq:r1} \\
      rd(wr(x, i, e),i) \rightarrow e &\hskip .6cm {\hbox{\rm for $i\in \tilde I$, $e\in \tilde E$}}  \label{eq:r2}\\
      wr(wr(x, i, e), j, d) \rightarrow wr(wr(x, j, d), i, e) &\hskip .6cm {\hbox{\rm for $i, j\in \tilde I$, $e, d\in \tilde E$, $i> j$}} \label{eq:r3} \\
      wr(wr(x, i, e), i, d) \rightarrow wr(x, i, d). &\hskip .6cm {\hbox{\rm for $i\in \tilde I$, $e, d\in\tilde  E$}}\label{eq:r4}
    \end{eqnarray}
    is confluent
     and ground
    irreducible;\footnote{The latter means that no rule can be used to
      reduce the left-hand or the right-hand side of another ground rule.
      Notice that
      ground rules from $A_R$ are precisely the rules obtained by
      orienting an equality from $A_M$
      (rules~\eqref{eq:r1}-\eqref{eq:r4} are not ground as they
      contain one \emph{variable}, namely the array variable $x$).}
  \item\noindent{\hskip-12 pt\rm (iv)}\
    if $a\uguale wr(b, I, E)\in A_M$ and $i, e$ are in the same position in
    the lists $I, E$, respectively, then $rd(b,i)\not\downarrow_{A_R}
    e$; 
  \item\noindent{\hskip-12 pt\rm\phantom{i}(v)}\ $\lbrace \diff(a,b)\uguale i, \diff(a', b')\uguale i'\rbrace
    \subseteq A_I$ and $a\downarrow_{A_R} a'$ and $ b\downarrow_{A_R}
    b'$ imply $i\coincide i'$;
  \item\noindent{\hskip-12 pt\rm (vi)}\ $\diff(a,b)\uguale i\in A_I$ and $rd(a,i)\downarrow_{A_R}
    rd(b,i)$ imply $a\downarrow_{A_R} b$.
  \end{desCription}
\end{definition}
Condition (o) means that the index constants occurring in a \nform
constraint are implicitly assumed to denote distinct objects.  This is
supported also by the statement of Lemma~\ref{lem:normal} below, from
which, it is evident that the addition of all the negative literals
$i\not\uguale j$ 
(for $i,j\in \tilde I,$ with $i\not\coincide j$) 
does
not compromise the satisfiability of a \nform constraint, precisely
because such negative literals are implicitly (already) part of the constraint.
In Condition (i), negative array literals $a\neq b$ are not allowed
because they can be replaced by suitable literals involving fresh
constants and the \diff\ operation (see axiom~\eqref{ax3a}).
%
%
%
Rules \eqref{eq:r1} and \eqref{eq:r2} mentioned in condition (iii)
reduce read-over-writes and rules \eqref{eq:r3} and \eqref{eq:r4} sort
indexes in flat terms $wr(a, I, E)$ in ascending order.  In addition,
condition (iv) prevents further redundancies in our rules.  Finally,
conditions (v) and (vi) deal with \diff.  In particular, (v) says that
\diff \ is ``well defined'' and (vi) is a ``conditional'' translation
of the contraposition of axiom~\eqref{ax3a}.

The non-ground rules from Definition~\ref{def:normal}(iii) form a
convergent rewrite system (critical pairs are confluent): this can be
checked manually (and can be confirmed also by 
tools like SPASS or MAUDE). 
Ground rules from $A_R$ are of the form
\begin{eqnarray}
&
a\rightarrow wr(b, I, E),
\label{eq:rules1} \\
&
rd(a, i)\rightarrow e,~~~~~
\label{eq:rules2}\\
&
\label{eq:rules3}
 e\rightarrow d.~~~~~~~~~~~~~~
\end{eqnarray}
Only rules of the form~\eqref{eq:rules3} can overlap with the
non-ground rules~\eqref{eq:r1}-\eqref{eq:r4}, but the resulting
critical pairs are trivially confluent.  Thus, in order to check
confluence of $A_M$, \emph{only overlaps between ground
  rules~\eqref{eq:rules1}-\eqref{eq:rules3} need to be considered}
(this is the main advantage of our choice to orient equalities
$a\uguale wr(b, I, E)$ from left to right instead of right to left).

\begin{lemma}\label{lem:normal}
Suppose that $A$ is \nform.
Then $A$  is \AXDIFF-satisfiable iff 
 there is no element
  inequality $e\neq d$ in $A_M$ such that $ e \downarrow_{A_R}  d$. 
 Moreover, $A$ is \AXDIFF-satisfiable iff 
  $$
  A\cup \{i\neq j\,\vert\, i,j\in \tilde I, i\not \coincide j\}\cup \{\alpha\neq \beta\}_{\alpha,\beta}
 $$ (varying $\alpha,\beta$
  among the different pairs of element and array constants in normal form occurring in $A$) is $\AXDIFF$-satisfiable.
%
\end{lemma}
\begin{proof}
  Clearly, the satisfiability of $A$ implies that for no negative
  index literal $e\not\uguale d$ from $A_M$, we have that
  $e\downarrow_{A_R} d$. Assume
 conversely that this is the case:
  our
  aim is to build a model for $A\cup \{\alpha\neq \beta\}_{\alpha,\beta}\cup \{i\neq j\}_{i,j}$ (varying $\alpha,\beta$
  and $i,j$ as indicated in the statement of the Lemma).
   We can freely make the following
  \emph{further assumption}: if $a, i$ occur in $A$ and $a$ is in
  normal form, there is some $e$ such that $rd(a,i)\uguale e$ belongs
  to $ A$ (in fact, if this does not hold, it is sufficient to add a
  further equality $rd(a,i)\uguale e$ - with fresh $e$ - without
  destroying the \nform property of the constraint).

  Let $I^*$ be the set of constants of sort $\INDEX$ occurring in $A$
  and let $E^*$ be the set of constants of sort $\ELEM$ in normal form
  occurring in $A$ (we have $I^*=\tilde I$ and $E^*\subseteq \tilde
  E$).
  Finally, we let $X$ be the set of free constants of sort \ARRAY
  occurring in $A$ which are in normal form.

  We build a model $\cM$ as follows (the symbol $+$ denotes disjoint
  union):
  \begin{iteMize}{$\bullet$} 
  \item ${\INDEX}^{\cM} \coincide I^*+\{*\}$; 
  \item ${\ELEM}^{\cM} \coincide E^*+X$; 
  \item $\ARRAY^\cM$ is the set of total functions from ${\INDEX}^{\cM}$ to
      ${\ELEM}^{\cM}$, $rd^\cM$ and $wr^\cM$ are the
    standard read and write operations (i.e. $rd^\cM$
      is  function application and $wr^\cM$ is the
      operation of modifying the first argument function 
      by giving it the third argument as a value for the second argument input);\footnote{ 
    In the terminology used in Section~\ref{subsec:semantic-arg}, this means that $\cM$ is a standard model.
    }
  \item for a constant $i$ of sort $\INDEX$, $i^\cM \coincide i$ for all $i\in
    I^*$; 
  \item for a constant $e$ of sort $\ELEM$, $e^\cM$ is the normal form
    of $e$;
  \item for a constant $a$ of sort $\ARRAY$ in normal form and $i\in
    I^*$, we put $a^\cM(i)$ to be equal
    to the normal form of $rd(a,i)$ (this is some $e\in \ELEM^\cM$ by
    our further assumption above); we also put $a^\cM(*) \coincide a$  
    (notice that ${\ELEM}^{\cM} \coincide E^*+X$, hence $a\in \ELEM^{\cM}$).
  \item for a constant $a$ of sort $\ARRAY$ not in normal form, let
    $wr(c, I, E)$ be the normal form of $a$: we let $a^\cM$ to be
    equal to $wr^\cM(c^\cM, I^\cM, E^\cM)$ ({This definition is
      correct because $a$ and $c$ cannot coincide; in fact, since
      $a<wr(a, I, E)$, the term $wr(a, I, E)$ cannot be the normal
      form of $a$.})
  \item we shall define $\diff^\cM$ later on.
  \end{iteMize}
  It is clear that in this way we have that all constants $\alpha$ of
  sort \ELEM or \ARRAY are interpreted in such a way that, if $\hat
  \alpha$ is the normal form of $\alpha$, then
  \begin{equation}
    \alpha^\cM= \hat{\alpha}^\cM.\label{eq:nfi}
  \end{equation}
  Also notice that, by the definition of $a^\cM$, if $e$ is the normal
  form of $rd(a,i)$, then we have
  \begin{equation}
    rd(a,i)^\cM= e^\cM\label{eq:nfi1}
  \end{equation}
  in any case (whether $a$ is in normal form or not). Finally, if $wr(c, I, E)$ is the normal form of $a$, then
 \begin{equation}\label{eq:nfi2}
  a^\cM = c^\cM \quad \Rightarrow \quad (I=\emptyset ~{\rm and}~E=\emptyset);
 \end{equation}
 this is because the only rule that can reduce $a$ must have $a$ as
 left-hand side and $wr(c,I, E)$ as right-hand side (rules are ground
 irreducible), thus in the rule $a\to wr(c, I, E)\in A_M$ we must have
 $I=\emptyset, E=\emptyset$ in case $a^\cM = c^\cM$ (recall
 Definition~\ref{def:normal}(iv)). 
 In more details, suppose that $I$ and $E$ are not empty and take
 $i\in I$ and $e\in E$ in corresponding positions.  We have that
 $rd(c,i)^\cM=rd^\cM(c^\cM, i^\cM)=
 c^\cM(i^\cM)=a^\cM(i^\cM)=rd^\cM(a^\cM, i^\cM)=rd(a,i)^\cM$ (we used
 the definition of interpretation of a ground term, the fact that
 $rd^\cM$ is interpreted as functional application and that
 $a^\cM=c^\cM$). Now, since $rd(a,i)$ normalizes to $e$,
 applying~(\ref{eq:nfi1}), we get that $rd(c,i)^\cM=e^\cM$, which
 means, again by~(\ref{eq:nfi1}), that $rd(c,i)$ normalizes to $e$ too
 ($e$ is in normal form, thus if $\tilde e$ is the normal form of
 $rd(c,i)$, we have that $\tilde e^\cM=e^\cM$ implies $e\coincide
 \tilde e$). This is contrary to Definition~\ref{def:normal}(iv). 

 Since $A$ is \nform, literals in $A$ are flat.  It is clear that all
 negative literals from $A$ are true: in fact, a \nform constraint
 does not contain inequalities between array constants, inequalities
 between index constants are true by construction and inequalities
 between element constants are true by the hypothesis of the
 Lemma. 
 Also, if $\alpha, \beta$ are either element or array constants in normal form, 
 we have $\alpha^\cM\neq \beta^\cM$ by construction (in particular, the interpretation of different array constants both in normal form differ at index $*$).
 Let us now consider positive literals in $A$: those from $A_M$
 are equalities of terms of sort $\ELEM$ or $\ARRAY$ and consequently
 are of the kind
 $$
 e\uguale d, \qquad a\uguale wr(c, I, E), \quad rd(a,i)\uguale e.
 $$ 
 Since ground rules are irreducible, $d$ is the normal form of $e$ and
 $wr(c, I, E)$ is the normal form of $a$, hence we have $e^\cM= d^\cM$
 and $a^\cM=wr(c, I, E)^\cM$ by~\eqref{eq:nfi} above. For the same
 reason $a$ and $e$ are in normal form in $rd(a,i)\uguale e$, hence
 $rd(a,i)^\cM= e^\cM$ follows by construction.

 It remains to define $\diff^\cM$ in such a way that flat literals
 $\diff(a,b)\uguale i$ from $A_I$ are true and the axiom~\eqref{ax3a}
 is satisfied.  Before doing that, let us observe that for all free
 constants $a,b$ occurring in $A$, \emph{we have that $a^\cM= b^\cM$
   is equivalent to $a\downarrow_{A_R} b$}. In fact, one side is
 by~\eqref{eq:nfi}; for the other side, suppose that $a^\cM= b^\cM$
 and that $wr(c, I, E)$, $wr(c', I', E')$ are the normal forms of $a$
 and $b$, respectively.  Then $c$ must be equal to $c'$, otherwise
 $a^\cM$ and $b^\cM$ would differ at index $*$.  If either $a$ or $b$
 is equal to $c$, trivially $a\downarrow_{A_R} b$ follows
 from~\eqref{eq:nfi2}.
 Otherwise, $a$ and $b$ are both reducible in $A_R$ and since
 ground rules are irreducible and the only rules that can reduce an
 array constant have the left-hand side equal to that array constant,
 we have that $a\to wr(c, I, E)$ and $b\to wr(c, I', E')$ are both
 rules in $A_R$: as such, they are subject to Condition (iv) from
 Definition~\ref{def:normal}.  First observe that we must have that
 $I\coincide I'$: otherwise, if there is $i\in I\setminus I'$, we
 could infer the following: (i) by ~\eqref{eq:nfi},
 $b^\cM(i)=c^\cM(i)$; (ii) $c^\cM(i)$ is the normal form of $rd(c,i)$
 by construction; (iii) by $a^\cM= b^\cM$, $c^\cM(i)$ is also equal to
 the normal form of the $e$ having in the list $E$ the same position
 as $i$ in the list $I$, contrary to Condition (iv) from
 Definition~\ref{def:normal}. Since terms are normalized with respect
 to rule~\eqref{eq:r3}, $I$ and $I'$ coincide not only as sets, but
 also as lists; this means that the lists $E$ and $E'$ coincide too
 (the terms $wr(c, I, E)$, $wr(c, I, E')$ are in normal form and we
 have $wr(c, I, E)^\cM=wr(c, I, E')^\cM$). In more details, let $i, e,
 \tilde e$ be in the $k$-th positions in the lists $I, E, E'$,
 respectively. From $wr(c, I, E)^\cM=wr(c, I, E')^\cM$, applying
 $rd^\cM(-, i^\cM)$, we get $e^\cM=\tilde e^\cM$, i.e. $e
 \downarrow_{A_R} \tilde e$, which means $e\coincide\tilde e$ because
 $wr(c, I, E)$, $wr(c, I, E')$ are in normal form (in particular,
 their sub-terms $e, \tilde e$ are not reducible). In conclusion,
 $a\downarrow_{A_R} b$ holds.

 Among the elements of $\ARRAY^\cM$, some of them are of the kind
 $a^\cM$ for some free constant $a$ of sort $\ARRAY$ occurring in $A$
 and some are not of this kind: we call the former `definable'
 arrays. In principle, it could be that $a^\cM=b^\cM$ for different
 $a, b$, but we have shown that this is possible only when $a$ and $b$
 have the same normal form.

 We are ready to define $\diff^\cM$: we must assign a value
 $\diff^\cM({\rm a},{\rm b})$ to all pairs of arrays ${\rm a}, {\rm
   b}\in \ARRAY^\cM$. If ${\rm a}$ or ${\rm b}$ is not definable or if
 there are no $a, b$ defining them such that $\diff(a,b)$ occurs in
 $A_I$, we can easily find $\diff^\cM({\rm a},{\rm b})$ so that axiom
 \eqref{ax3a} is true for ${\rm a}, {\rm b}$: one picks an index where
 they differ if they are not identical, otherwise the definition can
 be arbitrary. So let us concentrate into the case in which ${\rm a,
   b}$ are defined by constants $a, b$ such that the literal $\diff(a,
 b)\uguale i$ occurs in $A_I$: in this case, we define
 $\diff^\cM(a^\cM, b^\cM)$ to be $i$: Condition (v) from
 Definition~\ref{def:normal} (together with the above observation that
 two constants defining the same array in $\cM$ must have an identical
 normal form) ensures that the definition is correct and that all
 literals $\diff(a, b)\uguale i\in A_I$ becomes true. Finally,
 axiom~\eqref{ax3a} is satisfied by Condition (vi) from
 Definition~\ref{def:normal} and the fact that
 $rd(a,i)^\cM=rd(b,i)^\cM$ is equivalent to $rd(a,i)\downarrow_{A_R}
 rd(b,i)$ (to see the latter, just recall~\eqref{eq:nfi1}).
\end{proof}

\begin{remark}
  As we said, the importance of Definition~\ref{def:normal} lies in
  Lemma~\ref{lem:normal} and in Proposition~\ref{prop:merging}
  below. On the other hand, it is not true that if $A$ is \nform, then
  $A$ entails (modulo \AXDIFF) a positive literal $t\uguale v$ iff
  $t\downarrow_{A_R} v$, even in case $t, v$ are ground flat
  terms. {As a counterexample, consider $A=\{rd(a,i)\to e\}$; we have
    $A\vdash_{\AXDIFF} a\uguale wr(a,i,e)$ but $a\not\downarrow_{A_R}
    wr(a,i, e)$.  However, the proof of Lemma~\ref{lem:normal} shows
    that the following weaker---but still important---property holds:
    if $A$ is \nform and $t, v$ are terms of the same sort
    \emph{occurring in $A$}, then $A\vdash_{\AXDIFF} t\uguale v$ iff
    $t\downarrow_{A_R} v$.
%
%
 }  This may look unusual, however recall
  that our aim is not to decide equality by normalization but to have
   algorithms for satisfiability and interpolation.
\end{remark}

\subsection{Combining modular constraints}
Let $A, B$ be two constraints in the signatures $\Sigma^A, \Sigma^B$
obtained from the signature $\Sigma$ by adding some free constants and
let $\Sigma^C \coincide \Sigma^A \cap \Sigma^B$.  Given a term, a literal or
a formula $\varphi$ we call it:
\begin{iteMize}{$\bullet$}
\item {\em \abcommon} iff it is defined over $\Sigma^C$;
\item {\em \alocal} (resp. {\em \blocal}) if it is defined over
 $\Sigma^A$ (resp. $\Sigma^B$);
\item {\em $A$-strict} (resp. {\em $B$-strict}) iff it is \alocal
  (resp. \blocal) but not \abcommon;
\item {\em \abmixed} if it contains symbols in both $(\Sigma^A
  \setminus \Sigma^C)$ and $(\Sigma^B \setminus \Sigma^C)$;
\item {\em \abpure} if it does not contain symbols in both $(\Sigma^A
  \setminus \Sigma^C)$ and $(\Sigma^B \setminus \Sigma^C)$.
\end{iteMize}
(Notice that, sometimes in the literature about interpolation,
``\alocal'' and ``\blocal'' are used to denote what we call here
``$A$-strict'' and ``$B$-strict'').  The following modularity result
is crucial to
justify our interpolation algorithm for
\AXDIFF.
\begin{proposition}
  \label{prop:merging}
  Let $ A=\langle A_I,A_M \rangle$ and $B=\langle B_I,B_M \rangle$ be
  constraints in expanded signatures $\Sigma^A, \Sigma^B$ as above
  (here $\Sigma$ is the signature of \AXDIFF); let $A,B$ be both
  consistent and \nform.  Then $A\cup B$ is consistent and \nform, in
  case all the following conditions hold:
  \begin{desCription}
   \item\noindent{\hskip-12 pt\rm\phantom{I}(O)}\  an \abcommon literal belongs to $A$ iff it belongs to $B$; 
  \item\noindent{\hskip-12 pt\rm\phantom{II}(I)}\ every rewrite rule in $A_M\cup B_M$ whose left-hand side is
    \abcommon has also an \abcommon right-hand side;
  \item\noindent{\hskip-12 pt\rm\phantom{I}(II)}\ if $a, b$ are both \abcommon and $\diff(a,b)\uguale i \in
    A_I\cup B_I$, then $i$ is \abcommon too;
  \item\noindent{\hskip-12 pt\rm (III)}\  if a rewrite rule of the kind $a\rightarrow wr(c, I, E)$
    is in $A_M\cup B_M$ and the term $wr(c, I, E)$ is \abcommon, so
    is the constant $a$.
  \end{desCription}
\end{proposition}
\begin{proof}
  Since we cannot rewrite \abcommon terms to terms which are not, it
  is easy to see that $A_M\cup B_M$ is still convergent and ground
  irreducible; the other conditions from Definition~\ref{def:normal}
  are trivial, except condition (v). The latter is guaranteed by the
  hypotheses (II)-(III) as follows: the relevant case is when, say $\diff(a,b)\uguale i\in A_I$ is
  \alocal and $\diff(a',b')\uguale i'\in B_I$ is \blocal. If $a\downarrow
  a'$, since $A_M$ and $B_M$ are ground irreducible, we have that a
  single rewrite step reduces both $a$ and $a'$ to their normal form,
  that is we have
  $$
  a\rightarrow wr(c, I, E) \leftarrow a'.
  $$
  Now $wr(c, I, E)$ is \abcommon, because the rules $a\rightarrow
  wr(c, I, E), a'\rightarrow wr(c, I, E)$ are in $A_M$ and in $B_M$,
  respectively.  
 By hypothesis (III), we have that $a$ and $a'$ are
  \abcommon too; the same applies to $b, b'$ and hence to $i, i'$ by
  (II). Thus $\diff(a',b')\uguale i'$ is \abcommon and belongs to $A_I$,
  hence $i\coincide i'$ because $A$ is \nform.

  Since all conditions from Definition~\ref{def:normal} are satisfied,
  $A\cup B$ is \nform.
  Lemma~\ref{lem:normal} applies, thus yielding consistency.
\end{proof}
The above proof is so easy mainly because ground rewrite rules cannot
superpose with the non ground rewrite
rules~\eqref{eq:r1}-\eqref{eq:r4} (with the exception of the rewrite
rules $e\to d$, that may superpose but with trivially confluent
critical pairs): this is the main benefit of our choice of orienting
equalities $a\uguale wr(b, I, E)$ from left-to-right (and not from
right-to-left).

We conclude this section with a remark about the combination of
modular constraints in $\AXDIFF$ with constraints in other theories.
The theory \AXDIFF is stably infinite (in all its sorts) but
non-convex: this means that it is suitable for Nelson-Oppen
combination, but that disjunctions of equalities (not just equalities)
need to be propagated from an \AXDIFF-constraint, in case it is
involved in a combined problem. Actually, this does not happen for
\nform constraints, 
as it is shown by the statement of 
Lemma~\ref{lem:normal}.
In other words, no
disjunction of equalities needs to be propagated 
from a \nform constraint $A$  
and only equalities
that can be syntactically extracted from $A$ need to be propagated.

\section{A Solver for Arrays with \diff}
\label{sec:solver}

The first step towards the quantifier-free interpolation procedure for
\AXDIFF is the design of a satisfiability 
solver.  Although a solver for this theory can be easily derived from
existing solvers for \AX or \AXEXT, we need a specific algorithm from
which interpolants can be extracted.
To do this, Lemma~\ref{lem:normal} will play an important role by
allowing for the design of $\exists$-equivalence preserving
transformations that, once successively applied to a given constraint
$A$, will bring it to a consistent \nform constraint (if possible).
Failure of applying these transformations implies that $A$ is
unsatisfiable. 
In other words, the $\exists$-equivalence preserving transformations
will determine whether a {finite} constraint $A$ is satisfiable or not
by transforming it into a \nform $\exists$-equivalent constraint.

One of the key design choice underlying our transformations is to
separate the ``index'' part, that will be handled by guessing, of a
constraint from the ``array'' and ``elem'' parts, that will be subject
to rewriting.  Another important design decision is to distinguish a
\emph{preprocessing} and a \emph{completion} phase.  In the
preprocessing phase, besides flattening (see, e.g.,~\cite{ARR}) and
similar operations, a complete guessing of equalities/inequalities
among index constants will be performed.  Indeed, this guessing will
be realized by backtracking: if the completion phase will terminate in
a failure, another guessing has to be tried and unsatisfiability can
only be declared when all guessing fail.  The completion phase will
guarantee the confluence of the current rewriting system $A_R$,
recall Definition~\ref{def:normal}.  The confluence of $A_R$ is the
main requirement for a constraint to be \nform.

\subsection{Preprocessing}
\label{subsec:preprocessing}

The preprocessing phase consists 
of the following sequential
steps applied to our initial constraint $A$:
\begin{desCription}
\item\noindent{\hskip-12 pt\framebox{Step 1}}\ Flatten $A$, by replacing sub-terms with
  fresh constants and by adding the related defining equalities.

\item\noindent{\hskip-12 pt\framebox{Step 2}}\ Replace array inequalities $a\not\uguale b$
  by the following literals ($i, e, d$ are fresh)
  $$
  \diff(a,b)\uguale i, \quad rd(b, i)\uguale e, \quad rd(a,i)\uguale d, \quad d\not\uguale e.
  $$

\item\noindent{\hskip-12 pt\framebox{Step 3}}\ Guess a partition of index constants, i.e.,
  for any pair of indexes $i, j$ add either $i\uguale j$ or
  $i\not\uguale j$ (but not both of them); then remove the positive
  literals $i\uguale j$ by replacing $i$ by $j$ everywhere (if $i>j$
  according to the symbol precedence, otherwise replace $j$ by $i$);
  if an inconsistent literal $i\not\uguale i$ is produced, try with
  another guess (and if all guesses fail, report \texttt{unsat}).

\item\noindent{\hskip-12 pt\framebox{Step 4}}\ For all $a, i$ such that $rd(a,i)\uguale e$
  does not occur in the constraint, add such a literal $rd(a,i)\uguale
  e$ with fresh $e$.
\end{desCription}
At the end of the preprocessing phase, we get a finite set of flat
constraints; \emph{the disjunction of these constraints is
  $\exists$-equivalent to the original constraint}. For each of these
constraints, go to the completion phase: \emph{if the transformations
  below can be exhaustively applied (without failure) to at least one
  of the constraints, report} \texttt{sat}, \emph{otherwise report}
\texttt{unsat}. 
Failure can be caused by instructions (V) below.

The reason for inserting Step 4 above is just to simplify Orientation
and Gaussian completion below.  Notice that, even if rules
$rd(a,i)\rightarrow e$ can be removed during completion, the following
\textbf{invariant} is maintained: \emph{terms $rd(a,i)$ always reduce
  to constants of sort \ELEM.}

\subsection{Completion}
\label{subsec:gaussian}

The completion phase consists in various
transformations that should be non-de\-ter\-mi\-nis\-ti\-cal\-ly executed 
until  no rule  or a failure instruction applies. 
For clarity, we divide the transformations
into five groups.

\noindent
\textbf{(I) {Orientation}.} This group contains a single instruction:
 get rid of badly orientable equalities, by
  using the equivalences
\textbf{Refl}exivity and \textbf{Sym\-m}e\-try  of
Figure~\ref{fig:key-lemmas}; a badly orientable
  equality $a\uguale wr(b, I, E) $  (with $a<b$), 
after normalization of the term $wr(b, I, E) $ with respect to the non-grund rules $\eqref{eq:r3}-\eqref{eq:r4}$,
is replaced by an equality of the form
  $b\uguale wr(a, I, D)$ and
 by the equalities $rd(a, I)\uguale E$ 
  (all
  ``read literals'' required by the left-hand side of \textbf{Symm} comes from the above 
  invariant). A badly orientable equality $a\uguale wr(a, I, E) $ is removed and replaced 
  by read literals only (or by nothing if $I,E$ are empty).

\noindent
\textbf{(II) {Gaussian completion}.}  We now take care of the
confluence of $A_R$ (i.e., point (iii) of
Definition~\ref{def:normal}).  To this end, we consider all the
critical pairs that may arise among our rewriting
rules~\eqref{eq:rules1}-\eqref{eq:rules3} (recall that there is no
need to examine overlaps involving the non ground
rules~\eqref{eq:r1}-\eqref{eq:r4}).
To treat the relevant critical pairs, we combine standard Knuth-Bendix
completion for congruence closure with a specific method (``Gaussian
completion'') based on equivalences \textbf{Symm}etry,
\textbf{Trans}itivity and \textbf{Confl}ict of
Figure~\ref{fig:key-lemmas}.
%
  %
The critical pairs are listed below.
Two preliminary observations are in order.  First, we normalize a
critical pair by using $\rightarrow_*$ before recovering convergence
by adding a suitably oriented equality and removing the parent
equalities (the symbol $\rightarrow_*$ denotes the reflexive and
transitive closure of the rewrite relation $\rightarrow$ induced by
the rewrite rules $A_R\cup\{\eqref{eq:r1}-\eqref{eq:r4} \}$).  Second,
the provisos of all the equivalences in Figure~\ref{fig:key-lemmas}
used below (i.e., \textbf{Symm}, \textbf{Trans}, and \textbf{Confl})
are satisfied because of the pre-processing Step 3 above.

\begin{desCription}
\item\noindent{\hskip-12 pt (C1):}\ \framebox{ $wr(b_1, I_1, E_1)_{~*}\!\!\leftarrow
    wr(b'_1, I'_1, E'_1) \leftarrow a \rightarrow wr(b'_2, I'_2, E'_2)
    \rightarrow_* wr(b_2, I_2, E_2)$} \smallskip \\
  with $b_1>b_2$. We proceed in two steps.  First, we use \textbf{Symm} (from right to
  left) to replace the parent rule $a \rightarrow wr(b'_1, I'_1,
  E'_1)$ with
  \begin{eqnarray*}
    wr(a, I_1, F) \uguale b_1 \wedge rd(a,I_1) \uguale E_1
  \end{eqnarray*}
  for a suitable list $F$ of constants of sort \ELEM (notice that the
  equalities $rd(b_1, I_1) \uguale F$, which are required to apply
  \textbf{Symm}, are already available because terms of the form
  $rd(b_1, j)$ for $j$ in $I_1$ always reduce to constants of sort
  \ELEM by the invariant 
  resulting from the application of Step 4 in the pre-processing
  phase). Then, we apply \textbf{Trans} to the previously derived
  equality $b_1\uguale wr(a, I_1, F)$ and to the normalized second
  equality of the critical pair (i.e., $a \uguale wr(b_2, I_2, E_2)$)
  and we derive
  \begin{eqnarray}\label{eq:from_C1}
    b_1 \uguale wr(b_2,I_2\cdot I_1, E_2\cdot F) \wedge a \uguale wr(b_2,I_2,E_2) .
  \end{eqnarray}
  Hence, we are entitled to replace $b_1\uguale wr(a, I_1, F)$ with
  the rule $b_1\rightarrow wr(b_2, J, D)$, where $J$ and $D$ are lists
  obtained by normalizing the right-hand-side of the first equality
  of~\eqref{eq:from_C1} with respect to the non-ground
  rules~\eqref{eq:r3} and \eqref{eq:r4}.  To summarize: the parent
  rules are removed and replaced by the rules
  $$
  b_1\to wr(b_2, J,D), \quad 
  a \to wr(b_2,I_2,E_2) 
  $$
  and a bunch of new
  equalities of the form $rd(a,i)\uguale e$, giving rise, in turn, to
  rules of the form $rd(b_2,i) \to e$ or to rewrite rules of the
  form~\eqref{eq:rules3} after normalization of their left members
  (normalization of terms $rd(a,i)$ is indeed needed for the termination argument of Theorem~\ref{thm:axdiff_sat} below to work).
\item\noindent{\hskip-12 pt (C2):}\ 
  \framebox{ $ wr(b, I_1, E_1)_{~*}\!\!\leftarrow wr(b'_1, I'_1, E'_1)
    \leftarrow a \rightarrow wr(b'_2, I'_2, E'_2) \rightarrow_* wr(b,
    I_2, E_2)$} \smallskip \\
  Since identities like $wr(c, H, G)\uguale wr(c, \pi(H), \pi(G))$
  are \AXDIFF-valid for every permutation $\pi$ (under the proviso  $Distinct(H)$), 
  it is harmless to suppose that the 
  set of index variables $I \coincide I_1\cap I_2$ coincides 
  with the common prefix of  the lists $I_1$
  and $I_2$; hence we have  
   $I_1\coincide I\cdot J$ and $I_2\coincide I\cdot H$ for
  suitable disjoint lists $J$ and $H$.  Then, let $E$ and $E'$ be the prefixes
  of $E_1$ and $E_2$, respectively, of length equal to that of $I$;
  and let $E_1\coincide E\cdot D$ and $E_2\coincide E'\cdot F$ for suitable lists $D$
  and $F$.  At this point, we can apply \textbf{Confl} to replace both
  parent rules forming the critical pair with 
  \begin{eqnarray*}
    a\uguale wr(b,I,E) \wedge E\uguale E'\wedge rd(b,J) \uguale D\wedge rd(b,H)\uguale F ,
  \end{eqnarray*}
  where the first equality is oriented from left to right (i.e.,
  $a\rightarrow wr(b, I, E)$).
\end{desCription}

\noindent
\textbf{(III)  {Knuth-Bendix completion}.}
 The remaining critical pairs are treated by standard completion methods
for congruence closure.

\begin{desCription}
\item\noindent{\hskip-12 pt (C3):}\ 
  \framebox{ $ d_{~*}\!\!\leftarrow rd( wr(b, I, E),i) \leftarrow
    rd(a,i)\rightarrow e' \rightarrow_* e$} \smallskip \\
  Remove the parent rule $rd(a,i)\rightarrow e'$ and, depending on
  whether $d>e, e>d$, or $d\coincide e$, add the rule $d\to e$, $e\to
  d$, or do nothing.  (Notice that terms of the form $rd(b, j)$ are
  always reducible because of the invariant of Step 4 in the
  pre-processing phase; hence, $rd( wr(b, I, E),i)$ always reduces to
  some constant of sort $\ELEM$.)
\item\noindent{\hskip-12 pt (C4):}\ 
  \framebox{ $e_{~*}\!\!\leftarrow e' \leftarrow rd(a,i)\rightarrow
    d'\rightarrow_* d$} \smallskip \\ 
  Orient the critical pair (if $e$ and $d$ are
  not identical), add it as a new rule and remove one parent rule.
\item\noindent{\hskip-12 pt (C5):}\ 
  \framebox{ $d_{~*}\!\!\leftarrow d' \leftarrow e\rightarrow
    d'_1\rightarrow_* d_1$} \smallskip \\ 
  Orient the critical pair (if $d$ and $d_1$
  are not identical), add it as a new rule and remove one parent rule.
\end{desCription}

\noindent
\textbf{(IV) {Reduction}.}
The instructions in this group simplify the current rewrite rules.

\begin{desCription}
\item\noindent{\hskip-12 pt (R1):}\ If the right-hand side of a current ground
  rewrite rule  can
  be reduced, reduce it as much as possible, remove the old rule, and
  replace it with the newly obtained reduced rule. 
  Redundant equalities like $t=t$ are also removed.
\item\noindent{\hskip-12 pt (R2):}\ For every rule $a\rightarrow wr(b, I, E)\in A_M$,
 after normalization of the term $wr(b, I, E) $ with respect to the non-grund rules $\eqref{eq:r3}-\eqref{eq:r4}$,
   exhaustively
  apply \textbf{Red}uction in Figure~\ref{fig:key-lemmas} from left to right
  (this amounts to do the following: if there are $i, e$ in the same
  position $k$ in the lists $I, E$ such that
  $rd(b,i)\downarrow_{A_R} e$, replace $a\uguale wr(b, I, E)$ with
  $a\uguale wr(b, I\!-\!k, E\!-\!k)$).
\item\noindent{\hskip-12 pt (R3):}\
 If $\diff(a,b)\uguale i\in A_I$, $rd(a,i)\downarrow_{A_R}
  rd(b,i)$ and $a>b$, add the rule $a\to b$; replace also  $\diff(a,b)\uguale i$ by $\diff(b,b)\uguale i$
(this is needed for termination, it prevents the rule for being indefinitely applied).
\end{desCription}

\noindent
\textbf{(V) {Failure}.}
The instructions in this group aim at detecting inconsistency.

\begin{desCription}
\item\noindent{\hskip-12 pt (U1):}\ If for some negative literal $e\not\uguale d\in A_M$ we
  have $e\downarrow_{A_R} d$, report \texttt{failure} and backtrack to
  Step 3 of the pre-processing phase.
\item\noindent{\hskip-12 pt (U2):}\ If $\lbrace \diff(a,b)\uguale i, \diff(a', b')\uguale i'\rbrace
  \subseteq A_I$ and $a\downarrow_{A_R} a'$ and $ b\downarrow_{A_R}
  b'$ for $i\not\coincide i'$, report \texttt{failure} and backtrack to Step 3
  of the pre-processing phase.\smallskip
\end{desCription}

\noindent Notice that the instructions in the last two groups may require a
confluence test $\alpha\downarrow_{A_R}\beta$ that can be effectively
performed in case the instructions from groups (II)-(III) have been
exhaustively applied, because then all critical pairs have been
examined and the rewrite system $A_R$ is confluent.  If this is not
the case, one may pragmatically compute and compare any normal form of
$\alpha$ and $\beta$, keeping in mind that the test has to be repeated
when all completion instructions (II)-(III) have been 
exhaustively applied.

\begin{theorem}\label{thm:axdiff_sat}
  The above procedure decides constraint satisfiability in \AXDIFF.
\end{theorem}
\begin{proof}
  Correctness and completeness of the solver are clear: since all
  steps and instructions from Section~\ref{sec:solver} manipulate the
  constraint up to $\exists$-equivalence, it follows that if all
  guessings originated by Step 3 fail, the input constraint is
  unsatisfiable and, if one of them succeed, the exhaustive
  application of the completion instructions leads to a \nform
  constraint which is satisfiable by Lemma~\ref{lem:normal}.

  We must only consider termination; to show that any sequence of our
  instructions terminates, we use a standard technique. With every
  positive literal $l\uguale r$ we associate the multi-set of terms
  $\{l, r\}$; with every negative literal $l\not\uguale r$, we
  associate the multi-set of terms $\{ l, l, r,r\}$. Finally, with a
  constraint $A$ we associate the multi-set $M(A)$ of the multi-sets
  associated with every literal from $A$.  Now it is easy to see that
  such multi-set decreases after the application of any instruction.
\end{proof}
The termination analysis in the proof of Theorem~\ref{thm:axdiff_sat}
can be refined so as to show that our algorithm is in NP, which is
optimal because satisfiability of quantifier free formulae in \AXEXT
is already NP-complete~\cite{bradley}.

\section{The Interpolation Algorithm for Arrays with \diff}
\label{sec:isolver}

In the literature one can roughly distinguish two approaches to the
problem of computing interpolants.  In the former (see
e.g.~\cite{McM05,BKR+10}), an interpolating
calculus is obtained from a
standard calculus by adding decorations so as to enable the recursive
construction of an interpolating formula from a proof; in the latter
(see, e.g.,~\cite{YM05,KFG+09,CGS10}), the focus is on how to extend
an available decision procedure
to return interpolants.  
Our methodology is similar to the second approach, since we add the
capability of computing interpolants to the satisfiability procedure
in Section~\ref{sec:solver}.  However, we do this by designing a
flexible and abstract framework, relying on the identification of
\emph{basic operations} that can be performed independently
from the method used by the underlying satisfiability procedure to
derive a refutation.

\subsection{Interpolating Metarules}
\label{subsec:metarules}

Let now $A, B$ be constraints in signatures $\Sigma^A, \Sigma^B$
expanded with free constants and $\Sigma^C \coincide \Sigma^A \cap \Sigma^B$;
we shall refer to the definitions of \abcommon, \alocal, \blocal,
$A$-strict, $B$-strict,
\abmixed, \abpure terms, literals and formulae given in
Section~\ref{sec:interpolation}. Our goal is to produce, in case
$A\wedge B$ is \AXDIFF-unsatisfiable, a ground \abcommon sentence
$\phi$ such that $A\vdash_{ \AXDIFF} \phi$ and $\phi\wedge B$ is
\AXDIFF-unsatisfiable.

Let us 
examine some 
of the transformations to be applied to $A,B$. 
Suppose for instance that the literal $\psi$ is
\abcommon and such that $A\vdash_{\AXDIFF} \psi$; then we can 
transform $B$ into $B' \coincide B\cup \{\psi\}$. Suppose now that we got an
interpolant $\phi$ for the pair $A, B'$: clearly, we can
derive an interpolant for the original pair $A, B$ by taking
$\phi\wedge \psi$.  The idea is to collect some useful transformations
of this kind.  Notice that these transformations can also modify the
signatures $\Sigma^A, \Sigma^B$, 
in the sense that the signature of the pair $A',B'$ obtained after
applying a single transformation to a pair $A,B$ might be different from the signature of $A,B$ (typically, the
signature of $A', B'$ may contain extra fresh constants).
  For instance, suppose that $t$ is an
\abcommon term and that $c$ is a fresh constant; then we can put
$A' \coincide A\cup\{c=t\},\, B' \coincide B\cup\{c=t\}$: in fact, if $\phi$ is an
interpolant for $A', B'$, then $\phi(t/c)$ is an interpolant for $A,
B$. ({Notice that the fresh constant $c$ is now a shared symbol,
  because $\Sigma^A$ is enlarged to $\Sigma^A\cup\{c\}$, $\Sigma^B$ is
  enlarged to $\Sigma^B\cup\{c\}$ and hence $(\Sigma^A\cup\{c\})\cap
  (\Sigma^B\cup\{c\})=\Sigma^C\cup\{c\}$.}) The transformations we
need are called \emph{metarules} and are listed in
Table~\ref{tab:metarules} below (in the Table and more generally in
this Subsection, we use the notation $\phi\vdash \psi$ for
$\phi\vdash_{\AXDIFF} \psi$).\footnote{
Rules Redplus1, Redplus2 can be seen as instances of Rules Disjunction1, Disjunction2 (for $n=1$), thus they are redundant.
In Rule Propagate1, one can change the proviso to the weaker requirement `$\psi\in A$ and $\psi$ is \abcommon' (the case $A\vdash \psi$ could be 
obtained by applying Redplus1); a similar observation applies to Propagate2. We thank an anonymous referee
for these remarks.
}

An \emph{interpolating metarules refutation} for $A,B$ is a labelled
tree having the following properties: (i) nodes are labelled by pairs
of finite sets of constraints; (ii) the root is labelled by $A, B$;
(iii) the leaves are labelled by a pair $A, B$ such that $\bot \in
A\cup B$; (iv) each non-leaf node is the conclusion of a rule from
Table~\ref{tab:metarules} and its successors are the premises of that
rule.
The crucial properties of the metarules are summarized in the
following two Propositions.
\begin{proposition}
  \label{prop:equisat}  The unary metarules  
$ A\sep B\over{A'\sep B'}$
  from Table~\ref{tab:metarules} have the property that
  $A\wedge B$ is $\exists$-equivalent to $A'\wedge B'$;
  similarly, the
  $n$-ary metarules
  $
  A_1\sep B_1 
  ~~\cdots~~
  A_n\sep B_n\over{A\sep B}
  $
  from Table~\ref{tab:metarules} have the property that 
  $A\wedge B$ is $\exists$-equivalent to 
  $\bigvee_{k=1}^n(A_k\wedge B_k)$.
\end{proposition} 
%
%
\begin{proposition}
  \label{prop:metarules}
  If there exists an interpolating metarules refutation for $A, B$
  then there is a quantifier-free interpolant for $A, B$ (namely there
  exists a quantifier-free \abcommon sentence $\phi$ such that
  $A\vdash \phi$ and $B \wedge \phi \vdash \bot$). The interpolant
  $\phi$ is recursively computed applying the relevant interpolating
  instructions from Table~\ref{tab:metarules}.
\end{proposition}

%
\begin{table}[h!]
\scalebox{.95}{
\centering

\begin{tabular}{|cccccccccccc|}

\hline

\multicolumn{3}{|c|}{Close1} &
\multicolumn{3}{|c|}{Close2} &
\multicolumn{3}{|c|}{Propagate1} &
\multicolumn{3}{|c|}{Propagate2} \\

\hline


\multicolumn{3}{|c|}{
\begin{minipage}{.25\textwidth}
\begin{center}
\vspace{8pt}
\begin{prooftree}
\AxiomC{$~$}
\UnaryInfC{$A \sep B$}
\end{prooftree}
\vspace{8pt}
{\scriptsize
\begin{tabular}{rl}
\emph{Prv.}: & $A$ is unsat. \\
\emph{Int.}: & $\phi' \coincide \bot$.
\end{tabular}
}
\vspace{8pt}
\end{center}
\end{minipage}
}

&


\multicolumn{3}{|c|}{
\begin{minipage}{.25\textwidth}
\begin{center}
\vspace{8pt}
\begin{prooftree}
\AxiomC{$~$}
\UnaryInfC{$A \sep B$}
\end{prooftree}
\vspace{8pt}
{\scriptsize
\begin{tabular}{rl}
\emph{Prv.}: & $B$ is unsat. \\
\emph{Int.}: & $\phi' \coincide \top$.
\end{tabular}
}
\vspace{8pt}
\end{center}
\end{minipage}
}

&


\multicolumn{3}{|c|}{
\begin{minipage}{.25\textwidth}
\begin{center}
\vspace{8pt}
\begin{prooftree}
\AxiomC{$A \sep B\cup\{\psi\}$}
\UnaryInfC{$A \sep B$}
\end{prooftree}
\vspace{8pt}
{\scriptsize
\begin{tabular}{r@{ }l}
\emph{Prv.}: & $A\vdash \psi$ and \\ 
              & $\psi$ is $AB$-common. \\
\emph{Int.}: & $\phi' \coincide \phi \wedge \psi$.
\end{tabular}
}
\vspace{8pt}
\end{center}
\end{minipage}
}

&


\multicolumn{3}{|c|}{
\begin{minipage}{.25\textwidth}
\begin{center}

\begin{prooftree}
\AxiomC{$A\cup\{\psi\} \sep B$}
\UnaryInfC{$A \sep B$}
\end{prooftree}

\vspace{8pt}
{\scriptsize
\begin{tabular}{r@{ }l}
\emph{Prv.}: & $B\vdash \psi$ and \\
              & $\psi$ is $AB$-common. \\
\emph{Int.}: & $\phi' \coincide \psi \to\phi$.
\end{tabular}
}
\end{center}
\end{minipage}
}

\\
\hline
\hline

\multicolumn{4}{|c|}{Define0} &
\multicolumn{4}{|c|}{Define1} &
\multicolumn{4}{|c|}{Define2} \\

\hline


\multicolumn{4}{|c|}{
\begin{minipage}{.33\textwidth}
\begin{center}
\vspace{8pt}
\begin{prooftree}
\AxiomC{$A\cup\{a=t\} \sep B\cup\{a=t\}$}
\UnaryInfC{$A \sep B$}
\end{prooftree}
\vspace{8pt}
{\scriptsize
\begin{tabular}{r@{ }l}
\emph{Prv.}: & $t$ is $AB$-common, $a$ fresh. \\
\emph{Int.}: & $\phi' \coincide \phi(t/a)$.
\end{tabular}
}
\vspace{8pt}
\end{center}
\end{minipage}
}

&


\multicolumn{4}{|c|}{
\begin{minipage}{.33\textwidth}
\begin{center}
\begin{prooftree}
\AxiomC{$A\cup\{a=t\} \sep B$}
\UnaryInfC{$A \sep B$}
\end{prooftree}

\vspace{8pt}

{\scriptsize
\begin{tabular}{r@{ }l}
\emph{Prv.}: & $t$ is \alocal and $a$ is fresh. \\
\emph{Int.}: & $\phi' \coincide \phi$.
\end{tabular}
}
\end{center}
\end{minipage}
}

&


\multicolumn{4}{|c|}{
\begin{minipage}{.33\textwidth}
\begin{center}
\begin{prooftree}
\AxiomC{$A \sep B\cup\{a=t\}$}
\UnaryInfC{$A \sep B$}
\end{prooftree}

\vspace{8pt}

{\scriptsize
\begin{tabular}{rl}
\emph{Prv.}: & $t$ is \blocal and $a$ is fresh. \\
\emph{Int.}: & $\phi' \coincide \phi$.
\end{tabular}
}
\end{center}
\end{minipage}
}

\\
\hline
\hline

\multicolumn{6}{|c|}{Disjunction1} &
\multicolumn{6}{|c|}{Disjunction2} \\

\hline


\multicolumn{6}{|c|}{
\begin{minipage}{.5\textwidth}
\begin{center}
\vspace{8pt}
\begin{prooftree}
 \AxiomC{$\cdots ~~~A\cup\{\psi_k\} \sep B~~~\cdots~~~$}                                              
%
\UnaryInfC{$A \sep B$}                                                                                
\end{prooftree}
\vspace{8pt}
{\scriptsize
\begin{tabular}{rl}
\emph{Prv.}: & $\bigvee_{k=1}^n\psi_k$ is \alocal and $A \vdash \bigvee_{k=1}^n\psi_k$. \\
%
\emph{Int.}: & $\phi' \coincide \bigvee_{k=1}^n\phi_k$.
\end{tabular}
}
\vspace{8pt}
\end{center}
\end{minipage}
}

&


\multicolumn{6}{|c|}{
\begin{minipage}{.5\textwidth}
\begin{center}
\vspace{8pt}
\begin{prooftree}
\AxiomC{$\cdots ~~~A \sep B\cup\{\psi_k\}~~~\cdots~~~$}                     
%
\UnaryInfC{$A \sep B$}                                                       
\end{prooftree}
\vspace{8pt}
{\scriptsize
\begin{tabular}{rl}
\emph{Prv.}: & $\bigvee_{k=1}^n\psi_k$ is \blocal and $B \vdash \bigvee_{k=1}^n\psi_k$. \\
\emph{Int.}: & $\phi' \coincide \bigwedge_{k=1}^n\phi_k$.
\end{tabular}
}
\vspace{8pt}
\end{center}
\end{minipage}
}

\\
\hline
\hline

\multicolumn{3}{|c|}{Redplus1} &
\multicolumn{3}{|c|}{Redplus2} &
\multicolumn{3}{|c|}{Redminus1} &
\multicolumn{3}{|c|}{Redminus2} \\

\hline


\multicolumn{3}{|c|}{
\begin{minipage}{.25\textwidth}
\begin{center}
\vspace{8pt}
\begin{prooftree}
\AxiomC{$A \cup\{\psi\}\sep B$}
\UnaryInfC{$A \sep B$}
\end{prooftree}
\vspace{8pt}
{\scriptsize
\begin{tabular}{rl}
\emph{Prv.}: & $A\vdash \psi$ and \\ 
              & $\psi$ is \alocal. \\
\emph{Int.}: & $\phi' \coincide \phi$.
\end{tabular}
}
\vspace{8pt}
\end{center}
\end{minipage}
}

&


\multicolumn{3}{|c|}{
\begin{minipage}{.25\textwidth}
\begin{center}
\vspace{8pt}
\begin{prooftree}
\AxiomC{$A \sep B\cup\{\psi\}$}
\UnaryInfC{$A \sep B$}
\end{prooftree}
\vspace{8pt}
{\scriptsize
\begin{tabular}{rl}
\emph{Prv.}: & $B\vdash \psi$ and \\ 
              & $\psi$ is \blocal. \\
\emph{Int.}: & $\phi' \coincide \phi$.
\end{tabular}
}
\vspace{8pt}
\end{center}
\end{minipage}
}

&


\multicolumn{3}{|c|}{
\begin{minipage}{.25\textwidth}
\begin{center}
\vspace{8pt}
\begin{prooftree}
\AxiomC{$A \sep B$}
\UnaryInfC{$A \cup\{\psi\}\sep B$}
\end{prooftree}
\vspace{8pt}
{\scriptsize
\begin{tabular}{rl}
\emph{Prv.}: & $A\vdash \psi$ and \\
              & $\psi$ is \alocal. \\
\emph{Int.}: & $\phi' \coincide \phi$.
\end{tabular}
}
\vspace{8pt}
\end{center}
\end{minipage}
}

&


\multicolumn{3}{|c|}{
\begin{minipage}{.25\textwidth}
\begin{center}
\vspace{8pt}
\begin{prooftree}
\AxiomC{$A \sep B$}
\UnaryInfC{$A \sep B\cup\{\psi\}$}
\end{prooftree}
\vspace{8pt}
{\scriptsize
\begin{tabular}{rl}
\emph{Prv.}: & $B\vdash \psi$ and \\ 
              & $\psi$ is \blocal. \\
\emph{Int.}: & $\phi' \coincide \phi$.
\end{tabular}
}
\vspace{8pt}
\end{center}
\end{minipage}
}
\\

\hline
\hline

\multicolumn{4}{|c|}{ConstElim1} &
\multicolumn{4}{|c|}{ConstElim2} &
\multicolumn{4}{|c|}{ConstElim0} \\

\hline


\multicolumn{4}{|c|}{
\begin{minipage}{.33\textwidth}
\begin{center}
\vspace{8pt}
\begin{prooftree}
\AxiomC{$A \sep B$}
\UnaryInfC{$A\cup\{a=t\}\sep B$}
\end{prooftree}
\vspace{8pt}
{\scriptsize
\begin{tabular}{rl}
\emph{Prv.}: & $a$ is $A$-strict and \\ 
              & does not occur in $A, t$. \\
\emph{Int.}: & $\phi' \coincide \phi$.
\end{tabular}
}
\vspace{8pt}
\end{center}
\end{minipage}
}

&


\multicolumn{4}{|c|}{
\begin{minipage}{.33\textwidth}
\begin{center}
\vspace{8pt}
\begin{prooftree}
\AxiomC{$A \sep B$}
\UnaryInfC{$A \sep B\cup\{b=t\}$}
\end{prooftree}
\vspace{8pt}
{\scriptsize
\begin{tabular}{rl}
\emph{Prv.}: & $b$ is $B$-strict and \\ 
              & does not occur in $B, t$. \\
\emph{Int.}: & $\phi' \coincide \phi$.
\end{tabular}
}
\vspace{8pt}
\end{center}
\end{minipage}
}

&


\multicolumn{4}{|c|}{
\begin{minipage}{.33\textwidth}
\begin{center}
\vspace{8pt}
\begin{prooftree}
\AxiomC{$A \sep B$}
\UnaryInfC{$A\cup\{c=t\} \sep B\cup\{c=t\}$}
\end{prooftree}
\vspace{8pt}
{\scriptsize
\begin{tabular}{r@{ }l}
\emph{Prv.}: & $c$, $t$  are \abcommon, \\ 
              & $c$ does not occur in $A, B, t$. \\
\emph{Int.}: & $\phi' \coincide \phi$.
\end{tabular}
}
\vspace{8pt}
\end{center}
\end{minipage}
}

\\

\hline

\end{tabular}
}

\bigskip
\caption{\footnotesize{Interpolating Metarules:
each rule has a proviso $Prv.$ and an instruction $Int.$ for recursively computing the new interpolant $\phi'$
 from the old one(s) $\phi, \phi_1, \dots, \phi_k$.
  }
}
\label{tab:metarules}
\end{table}

The proofs of both Propositions~\ref{prop:equisat}
and~\ref{prop:metarules} are straightforward.  The following
observations are the basis of such proofs.  The metarules are applied
\textbf{bottom-up} whereas interpolants are computed (from an
interpolating refutation) in a \textbf{top-down} manner.  We should
have labelled nodes in an interpolating metarules refutation by
4-tuples $(\Sigma^A, A, \Sigma^B, B)$, where $\Sigma^A, \Sigma^B$ are
signatures expanded with free constants, $A$ is a
$\Sigma^A$-constraint and $B$ is a $\Sigma^B$-constraint. The
\emph{shared signature} of the node labelled $(\Sigma^A, A, \Sigma^B,
B)$ (i.e. the signature where interpolants are recursively computed)
is taken to be $\Sigma^C \coincide \Sigma^A\cap \Sigma^B$; the \emph{root
  signature pair} is the pair of signatures comprising all symbols
occurring in the original pair of constraints.  We did not make all
this explicit in order to avoid notation overhead. Notice that the
only metarules that modify the signatures are (Define0), (Define1),
(Define2) (which add $a$ to $\Sigma^A\cap \Sigma^B, \Sigma^A,
\Sigma^B$, respectively).  Some other rules like (ConstElim0),
(ConstElim1), (ConstElim2) could in principle restrict the signature,
but signature restriction is not relevant for the computation of
interpolants: there is no need that \emph{all} \abcommon symbols occur
in the interpolants, but we certainly do not want \emph{extra} symbols
to occur in them, so only bottom-up signature expansion must be
tracked.

\subsection{The Interpolating Solver
}
\label{subsec:isolver}

The metarules are complete, i.e.
if $A\wedge B$ is \AXDIFF-unsatisfiable, then
 (since we  know
that an interpolant exists) a
single application of (Propagate1) and (Close2) gives an interpolating
metarules refutation. This observation shows that
metarules are by no
means better than the
brute force enumeration of formulae to find interpolants. 
However, metarules are useful to design an algorithm manipulating
pairs of constraints based on transformation instructions.  In fact,
each of the transformation instructions can be \emph{justified} by a
metarule (or by a sequence of metarules): in this way, if our
instructions form a complete and terminating algorithm, we can use
Proposition~\ref{prop:metarules} to get the desired interpolants.  
The main advantage of using metarules as justifications is that we
just need to take care of the
completeness and
  termination
of the algorithm, 
and not
about interpolants anymore.  Here ``completeness'' means that our
transformations should be able to bring a pair $(A,B)$ 
of constraints into a pair $(A', B')$ that either matches the
requirements of Proposition~\ref{prop:merging} or is explicitly
inconsistent, in the sense that $\bot\in A'\cup B'$.  The latter is
obviously the case whenever the original pair $(A, B)$ is
\AXDIFF-unsatisfiable
and it is precisely the case leading to an interpolating metarules
refutation.

The basic idea is that of invoking the algorithm of
Section~\ref{sec:solver} on $A$ and $B$ separately and to propagate
equalities involving \abcommon terms.  We shall assume \emph{an
  ordering precedence making \abcommon constants smaller than
  $A$-strict or $B$-strict constants of the same sort}. However, this
is not sufficient to prevent the algorithm of Section~\ref{sec:solver}
from generating literals and rules violating one or more of the
hypotheses of Proposition~\ref{prop:merging}: this is why the extra
correcting instructions of group ($\gamma$) below are needed.  Our
interpolating algorithm has a pre-processing and a completion phase, like
the algorithm from
Section~\ref{sec:solver}.
%
%
%
 
\vskip 2mm
\noindent
{\textbf{Pre-processing.}}
In this phase the four Steps of Section~\ref{subsec:preprocessing} are
performed on both $A$ and  $B$; to justify these steps we need
metarules (Define0,1,2), (Redplus1,2), (Redminus1,2),
(Disjunction1,2), (ConstElim0,1,2), and (Propagate1,2)---the latter
because if $i,j$ are \abcommon, the guessing of $i\uguale j$ versus
$i\not\uguale j$ in Step 3 can be done, say, in the $A$-component and
then propagated to the $B$-component.
At the end of the
preprocessing phase, the following 
properties (to be maintained as invariants afterwards) hold:
\begin{enumerate}[({i}1):]
\item $A$ (resp. $B$) contains $i \not\uguale j$ for
  all \alocal (resp. \blocal) constants $i,j$ of sort \INDEX occurring
  in $A$ (resp. in $B$);
\item if $a, i$ occur in $A$ (resp. in $B$),
 then $rd(a,i)$ reduces to an \alocal (resp. \blocal) constant of sort
 \ELEM.
\end{enumerate}

\noindent
{\textbf{Completion.}} 
Some groups of instructions to be executed non-deterministically
constitute the completion phase.  There is however an important
difference here with respect to the completion phase of
Section~\ref{subsec:gaussian}: it may happen that we need some
\emph{guessing} also inside the completion phase (only
the instructions from group ($\gamma$) below may need such guessings).
Each instruction can be easily justified by suitable metarules 
(we omit the 
straightforward
details). 
The groups of instructions
are the following:
\begin{enumerate}[($\alpha$)]
 \item Apply to $A$ or to $B$ any instruction from the completion phase of Section~\ref{subsec:gaussian}.
\item[($\beta$)] If there is an \abcommon literal that belongs to $A$  but not to $B$ (or vice versa), copy it in $B$ (resp. in $A$).
\item[($\gamma$)] Replace \emph{undesired literals}, i.e., those
violating conditions (I)-(II)-(III) from Proposition~\ref{prop:merging}.
\end{enumerate}
To avoid trivial infinite loops with the $(\beta)$ instructions,
rules in $(\alpha)$ deleting an \abcommon literal should be performed \emph{simultaneously} 
in the $A$- and in the $B$-components (it can be easily checked - see the 
proof of Theorem~\ref{thm:main} below -
%
that this is always possible, 
\emph{if  rules in $(\beta)$ and $(\gamma)$ are given higher priority}).

Instructions ($\gamma$) need to be 
described in more details.
Preliminarily, we introduce 
a technique
that we call \emph{Term Sharing}.  Suppose that the $A$-component
contains a literal $\alpha =t$, where the term $t$ is \abcommon but
the free constant $\alpha$ is only \alocal. Then it is possible to
``make $\alpha$ \abcommon'' in the following way. First, introduce a
fresh \abcommon constant $\alpha'$ with the explicit definition
$\alpha'=t$ (to be inserted both in $A$ and in $B$, as justified by
metarule (Define0)); then replace the literal $\alpha=t$ by
$\alpha=\alpha'$ and replace $\alpha$ by $\alpha'$ everywhere else in
$A$; finally,
delete $\alpha\uguale\alpha'$ too.  The  
result is a pair
$(A, B)$ where basically nothing has changed but $\alpha$ has been
renamed to an \abcommon constant $\alpha'$. Notice that the above
transformations can be justified by metarules (Define0), (Redplus1),
(Redminus1), (ConstElim1).  We are now ready to explain instructions
($\gamma$) in details.  First, consider undesired literals
corresponding to the rewrite rules of the form
\begin{equation}
  \label{eq:undesired1}
  rd(c, i) \rightarrow d
\end{equation}
in which the left-hand side is \abcommon and the right-hand side is,
say, 
$A$-strict.
If we apply Term Sharing, we can solve
the problem by renaming $d$ to an \abcommon fresh constant $d'$.
We can apply a similar procedure to the rewrite rules
\begin{equation}
  \label{eq:undesired2}
  a \rightarrow wr(c, I, E) 
\end{equation}
in case the right-hand side is \abcommon and the left-hand side is not;
when we rename $a$ to some fresh \abcommon constant $c'$, we must
arrange the precedence so that $c'>c$ 
to orient the renamed
literal as $c'\rightarrow wr(c, I, E)$.
Then, consider the literals of the form
\begin{equation}
  \label{eq:undesired3}
  \diff(a,b)=k
\end{equation}
in which the left-hand side is \abcommon and the right-hand side is,
say, 
$A$-strict.
Again, we can rename $k$ to some
\abcommon constant 
$k'$
 by Term Sharing.  Notice that 
$k'$
is \abcommon, whereas $k$ was only \alocal: this implies that we might
need to perform some guessing to maintain the invariant
(i1). Basically, we need to repeat Step 3 from
Section~\ref{subsec:preprocessing} till invariant (i1) is restored
($k'$
must be compared for equality with the other \blocal
constants of 
sort \INDEX).
The last undesired literals to take care of are the rules of the form
\begin{equation}
  \label{eq:unsuitable}
  c\to wr(c', I, E)
\end{equation}
having an \abcommon left-hand side but, say, only an \alocal
right-hand side ({literals of the form $d=e$ are automatically
  oriented in the right way by our choice of the precedence}).
Notice that from the fact that $c$ is \abcommon, it follows (by our
choice of the precedence) that $c'$ is \abcommon too.  We can freely
suppose that $I$ and $E$ are
split into sub-lists $I_1, I_2$ and $E_1, E_2$, respectively, such that
$I\coincide I_1\cdot I_2$ and $E\coincide E_1\cdot E_2$,
where 
$I_1, E_1$ are \abcommon, $I_2\coincide i_1,\dots, i_n$,
$E_2\coincide e_1, \dots, e_n$ and for each $k=1, \dots, n$ at least
one from $i_k, e_k$ is not \abcommon. This $n$ (measuring essentially the number of non \abcommon symbols in~\eqref{eq:unsuitable}) is  called the
\emph{degree} of the undesired literal~\eqref{eq:unsuitable}:  
in the following, we shall see how 
to eliminate~\eqref{eq:unsuitable} or to
 replace it with a smaller degree literal. 
We first make
a guess (see metarule (Disjunction1)) about the truth value of the
literal $c\uguale wr(c', I_1, E_1)$.  In the first
case, we add the positive literal to the current constraint; as a
consequence, we get that the literal~\eqref{eq:unsuitable} is
equivalent to $c\uguale wr(c, I_2, E_2)$ and also to $rd(c,I_2)\uguale
E_2$ (see \textbf{Red} in Figure~\ref{fig:key-lemmas}).  In
conclusion, in this case, the literal~\eqref{eq:unsuitable} is
replaced by the \abcommon rewrite rule $c\to wr(c', I_1, E_1)$ and by
the literals $rd(c,I_2)\uguale E_2$.
In the second case, we guess that the negative literal $c\not\uguale
wr(c', I_1, E_1)$ holds; we introduce a fresh \abcommon constant $c''$
together with the
 defining \abcommon literal\footnote{We put $c>c''>c'$ in the
  precedence.  Notice that invariant (i2) is maintained, because all
  terms $rd(c'', h)$  normalize to an element constant. 
In case $I_1$ is empty, 
one can directly take $c'$ as $c''$.
 }
\begin{equation}
  \label{eq:newrew}
  c''\to wr(c', I_1, E_1)
\end{equation}
(see metarule (Define0)).
The literal~\eqref{eq:unsuitable} is replaced by the  literal
\begin{equation}
  \label{eq:diffnew1}
  c\to wr(c'', I_2, E_2).
\end{equation}
We show how to make the degree of~\eqref{eq:diffnew1} smaller than
$n$. In addition, we eliminate the negative literal $c\not\uguale c''$
coming from our guessing (notice that, according to~\eqref{eq:newrew}, $c''$ renames $wr(c', I_1, E_1)$).
This is done as follows:
we introduce
 fresh \abcommon constants 
 $i, d, d''$
together with the \abcommon defining literals
\begin{equation}
  \label{eq:diffnew}
  \diff(c, c'')\uguale i,\quad rd(c, i)\rightarrow d, \quad rd(c'', i) \rightarrow d''
\end{equation}
(see metarule (Define0)).
Now it is possible to replace $c\not \uguale c''$ by the literal $d\not\uguale d''$
(see axiom~\eqref{ax3a}). 
 Under the assumption  $Distinct(I_2)$,  the following statement is \AXDIFF valid:
$$
c\uguale wr(c'', I_2, E_2) 
\wedge  rd(c'', i)\uguale d'' \wedge rd(c, i)\uguale d \wedge d\not\uguale d''\to \bigvee_{k=1}^n (i\uguale i_k\wedge d\uguale e_k).
$$
Thus, we get $n$ alternatives (see metarule (Disjunction1)). In the $k$-th
alternative, we can remove 
the constants $i_k, e_k$ from the constraint, by replacing them with
the \abcommon terms $i, d$ respectively (see metarules (Redplus1), (Redplus2), (Redminus1), (Redminus2),(ConstElim1),(ConstElim0));
notice that it might be necessary to complete the index partition.
In this way, the degree of~\eqref{eq:diffnew1} is now smaller than $n$.


\emph{In conclusion,} if we apply exhaustively Pre-Processing and
Completion instructions above, starting from an initial pair of
constraints $(A, B)$, we can produce a tree, whose nodes are labelled
by pairs of constraints (the successor nodes of a node labelled
$(\tilde A, \tilde B)$ are labelled by pairs of constraints that are
obtained from $(\tilde A, \tilde B)$ by applying an
instruction). {Notice that the branching in the tree is due to
  instructions that need guessing and that Pre-Processing instructions
  are applied only in the initial segment of a branch.} We call such a
tree an \emph{in\-ter\-po\-la\-ting tree} for $(A, B)$. The following
result shows that we obtained an interpolation algorithm for \AXDIFF.
\begin{theorem}\label{thm:main}
  Any interpolation tree for $(A, B)$ is finite; moreover, it is an
  interpolating metarules refutation (from which an interpolant can be
  recursively computed according to Proposition~\ref{prop:metarules})
  precisely iff $A\wedge B$ is \AXDIFF-unsatisfiable.
\end{theorem}
\begin{proof}
 Since all instructions can be justified by metarules and since our instructions bring any pair of constraints into constraints which are either manifestly 
inconsistent (i.e. contain $\bot$) or satisfy the requirements of Proposition~\ref{prop:merging}, the second part of the claim is clear.
We only have to show that all branches are finite (then K\"onig lemma applies).

A complication that we may face here is due to the fact that during instructions ($\gamma$), the signature is enlarged. However,
  notice that our instructions may
introduce genuinely new \abcommon array constants, however \emph{they
  can only rename index constants, element constants and non \abcommon
  array constants}. Moreover: (1) Term Sharing decreases the number of
the constants which are not \abcommon;  (2) each call in the
recursive procedure for the elimination of literals~\eqref{eq:unsuitable},  \emph{either} (2.i) renames
to \abcommon constants some constants which were not \abcommon
before, \emph{or} (2.ii) just replaces a literal of the kind 
$c\uguale wr(c', I_1\cdot I_2, E_1\cdot E_2)$ by the literals
$$
c\uguale wr(c', I_1, E_1), \qquad rd(c', I_2)=E_2
$$
(see the first alternative following the guessing about truth of the literal $c\uguale wr(c', I_1, E_1)$).
Since there are only finitely many  non \abcommon constants at all, after finitely many steps neither Term Sharing nor (2.i)
apply anymore. 
 We finally show that instructions 
($\alpha$), ($\beta$) and (2.ii) (that do \emph{not} enlarge the signature) cannot be executed infinitely many times either. To this aim, it is sufficient to associate with each pair of constraints
$(\tilde A, \tilde B)$ the complexity measure given by the multi-set of pairs (ordered lexicographically) $\langle m(L), N_L\rangle$
(varying $L\in \tilde A\cup\tilde B$), where $m(L)$ is the multi-set of terms associated with the literal $L$ and $N_L$ is 1 if $L\in \tilde A\setminus\tilde B$, 
2 if
$L\in \tilde B\setminus\tilde  A$, and 0 if $L\in \tilde A\cap\tilde B$.
 In fact,  the second component in the above pairs takes care of instructions ($\beta$), whereas the first component covers all the remaining instructions.
Notice that it is important that, whenever an \abcommon literal is
deleted, the deletion happens simultaneously in both components
({otherwise, the ($\beta$) instruction could re-introduce it, causing
  an infinite loop; our complexity measure does not decrease if an
  \abcommon literal is replaced by smaller literals only in the $A$-
  or in the $B$-component}): in fact, it can be shown (by inspecting
the instructions from the completion phase of
Subsection~\ref{subsec:gaussian}) that whenever an \abcommon literal
is deleted, the instruction that removes it involves only \abcommon
literals, if undesired literals are removed first.\footnote{Let us see
  an example by considering instruction (C3). This instruction removes
  a literal $rd(a,i)\to e'$ using a literal $a\to wr(b,I,E)$ (and
  possibly rewrite rules $rd(b,i)\to d'$ as well as rewrite rules that
  might reduce some of the $e', d', E$). Now, if $rd(a,i)\to e'$ is
  \abcommon and all the other involved rules are not undesired
  literals, the instruction as a whole manipulates \abcommon
  literals. As such, if ($\beta$) has been conveniently applied, the
  instruction can be performed 
  simultaneously in the $A$- and in the
  $B$-component and our specification is precisely to do that.
} Thus, if instructions in ($\beta$) and ($\gamma$) have priority (as
required by our specifications 
above),
\abcommon literal deletions caused by ($\alpha$) can be performed both
in the $A$- and in the $B$-component (notice also that the
instructions from ($\beta$) and (2ii) do not remove \abcommon
literals).
\end{proof}

From the theorem above it immediately follows
Theorem~\ref{thm:amalgamation}, that we have already proved in
Section~\ref{subsec:semantic-arg} by using model-theoretic notions
(thus in a non-constructive way).

\subsection{An Example}

To illustrate our method, we describe the computation of an interpolant 
for the problem 
$$\Pi \coincide (A_0,\ B_0)$$  
where
\begin{eqnarray*}
A_0 & \coincide & \{\ a=wr(b,i,d)\ \} \\
B_0 & \coincide & \{\ rd(a,j)\not=rd(b,j),\ rd(a,k)\not=rd(b,k),\ j\not=k\ \}.
\end{eqnarray*}

Notice that $i, d$ are $A$-strict constants, $j, k$ are $B$-strict
constants, and $a, b$ are \abcommon constants with precedence
$a>b$. 
The computation of the interpolant in our framework can be represented with a tree, 
growing upward from $\Pi$, in which each step can be identified with a set of appropriate 
metarules application. 

To begin with we first apply Pre-Processing instructions to obtain 
\begin{eqnarray*}
A_1 & \coincide & \{\ a=wr(b,i,d),\ rd(a,i)=e_5,\ rd(b,i)=e_6\ \} \\
B_1 & \coincide & \{\, rd(a,j)=e_1
                  ,\, rd(b,j)=e_2
	          ,\, rd(a,k)=e_3
	          ,\, rd(b,k)=e_4
	          ,\, e_1\not=e_2
	          ,\, e_3\not=e_4
	          ,\, j\not=k\, \}.
\end{eqnarray*}
Since $a=wr(b,i,d)$ is an undesired literal of the
kind~\eqref{eq:unsuitable}, we 
generate the two sub-problems 
\begin{eqnarray*}
\Pi_1 & \coincide & (A_1 \cup \{\, rd(b,i)=d,\, a=b\, \},\ B_1), \mbox{ and } \\
\Pi_2 & \coincide & (A_1 \cup \{\, a \not= b\, \},\ B_1)
\end{eqnarray*}
{(this
is precisely the case in which there is no need of an extra \abcommon
constant $c''$).}  

Let us consider $\Pi_1$ first. 
Notice that $A \vdash a=b$, and $a=b$ is \abcommon. 
Therefore we send $a=b$ to $B_1$, and we may derive the new equality $e_1=e_2$ from 
the critical pair (C3) $e_1 \leftarrow rd(a,j) \rightarrow rd(b,j) \rightarrow e_2$,
thus obtaining 
\begin{eqnarray*}
A_2 & \coincide & \{\ rd(b,i)=d,\ a=b,\ rd(a,i)=e_5,\ rd(b,i)=e_6\ \} \\
B_2 & \coincide & \{\, rd(b,j)=e_2,\, rd(a,k)=e_3,\, rd(b,k)=e_4,\, e_1\not=e_2,\, e_3\not=e_4,\, j\not=k,\, a=b,\, e_1=e_2\, \}.
\end{eqnarray*}
Now $B$ is inconsistent (as it contains both $e_1\not=e_2$ and $e_1=e_2$). 
The interpolant for $\Pi_1$ can be computed with the 
{\em interpolating instructions} of the metarules 
(Close2, Redplus2, Redmius2, Propagate1) 
resulting in 
$$
\varphi_1 \coincide a=b
$$
as shown in Figure~\ref{fig:derivation_pi_1}.

\begin{figure}[h!]

\fbox{
\begin{minipage}{\textwidth}
\scriptsize
\medskip
\begin{prooftree}

\def\extraVskip{5pt}
\def\labelSpacing{10pt}
\def\fCenter{\sep}

\AxiomC{}
\LeftLabel{Close2}
\RightLabel{$\top$}
\UnaryInf$\quad\quad\quad\quad \ldots \quad\quad\quad\quad \fCenter B_1' \cup \{\, a=b,\, e_1=e_2\, \}$
\LeftLabel{Redminus2}
\RightLabel{$\top$}
\UnaryInf$\quad\quad\quad\quad \ldots \quad\quad\quad\quad \fCenter B_1 \cup \{\, a=b,\, e_1=e_2\, \}$
\LeftLabel{Redplus2}
\RightLabel{$\top$}
\UnaryInf$\quad\quad\quad\quad \ldots \quad\quad\quad\quad \fCenter B_1 \cup \{\, a=b\, \}$
\LeftLabel{Propagate1}
\RightLabel{$a=b$}
\UnaryInf$A_1 \cup \{\, rd(b,i)=d,\, a=b\, \} \fCenter B_1$
\end{prooftree}
\medskip

where \\ 
$B_1' \equiv B_1 \setminus \{ rd( b, j ) = e_2 \}$

\end{minipage}
}

\caption{Interpolant derivation for $\Pi_1$ using metarules. The
derivation is to be read bottom-up. The labels
for the rules are shown on the left, while the partial interpolants,
computed top-down,
are shown on the right.}
\label{fig:derivation_pi_1}

\end{figure}

Then, let us consider branch $\Pi_2$. Recall that this branch
originates from the attempt of removing the undesired rule 
$a \rightarrow wr(b,i,d)$.  
We introduce, in both $A$ and $B$, the \abcommon defining literals
$\diff(a,b)=l, rd(a,l)=f_1, rd(b,l)=f_2$. In order
to remove $a\not=b$, we introduce $f_1\not=f_2$ in $A$, which is propagated to $B$, 
thus obtaining:
\begin{eqnarray*}
A_3 & \coincide\ \{ &\! a = wr(b,i,d), \\ 
  &               &\! \diff(a,b)=l,\ rd(a,l)=f_1,\ rd(b,l)=f_2,\ f_1\not=f_2\ \} \\
B_3 & \coincide\ \{ &\! rd(a,j)=e_1,\, rd(b,j)=e_2,\, rd(a,k)=e_3,\, rd(b,k)=e_4,\, \\
  &               &\! e_1\not=e_2,\, e_3\not=e_4,\, j\not=k,\, \\
  &               &\! \diff(a,b)=l,\, rd(a,l)=f_1,\, rd(b,l)=f_2,\, f_1\not=f_2\ \}.
\end{eqnarray*}
Since $a=wr(b,i,d)$ contains only the index $i$, we do not have a real case 
split. Therefore we replace $i$ with $l$, and $d$ with $f_1$. 
At last, we propagate the \abcommon literal $a=wr(b,l,f_1)$ to $B$.
After all these steps we obtain:
\begin{eqnarray*}
A_4 & \coincide\ \{ &\! a=wr(b,l,f_1), \\
  &               &\! \diff(a,b)=l,\ rd(a,l)=f_1,\ rd(b,l)=f_2,\ f_1\not=f_2\ \} \\
B_4 & \coincide\ \{ &\! rd(a,j)=e_1,\ rd(b,j)=e_2,\ rd(a,k)=e_3,\ rd(b,k)=e_4, \\ 
  &               &\! e_1\not=e_2,\ e_3\not=e_4,\ j\not=k,\ \\
  &               &\! \diff(a,b)=l,\ rd(a,l)=f_1,\ rd(b,l)=f_2,\ f_1\not=f_2, \\ 
  &               &\! a=wr(b,l,f_1)\ \}.
\end{eqnarray*}
Since we have one more \abcommon index constant $l$, we
complete the current index constant partition, namely $\{ k \}$ 
and $\{ j \}$: we have three alternatives, to let $l$ stay alone 
in a new class, or to add $l$ to one of the two existing classes.
In the first alternative, because of the following critical pair (C3)
$e_1 \leftarrow rd(a,j) \rightarrow rd(wr(b,l,f_1),j) \rightarrow e_2$, we add
$e_1 = e_2$ to $B$, which becomes trivially unsatisfiable. The other
two alternatives yield similar outcomes.
\COMMENT{ 
For each sub-problem the interpolant, reconstructed by reverse
application of the interpolating instructions of (Define0) and (Propagate1),
is 
$$\varphi'_2 \coincide 
(a=wr(b,\diff(a,b),rd(a,\diff(a,b))) \wedge rd(a, \diff(a,b))\neq rd(b, \diff(a,b))).$$
The interpolant $\varphi_2$ for the branch $\Pi_2$
has to be computed by combining with (Disjunction2) 
three copies of $\varphi'_2$, and so 
$\varphi_2 \coincide \varphi'_2$. 
}
For each sub-problem the interpolant is $\top$. The partial interpolant for $\Pi_2$
has to be reconstructed by the reverse application of the interpolanting instructions
of (Define0) and (Propagate1), as shown in Figure~\ref{fig:derivation_pi_2}, which
yield 
$$
\varphi_2 \coincide 
(a=wr(b,\diff(a,b),rd(a,\diff(a,b))) \wedge rd(a, \diff(a,b))\neq rd(b, \diff(a,b))).
$$

\begin{figure}[h!]

\fbox{
\begin{minipage}{\textwidth}
\scriptsize
\medskip
\begin{prooftree}

\def\extraVskip{5pt}
\def\labelSpacing{10pt}
\def\fCenter{\sep}

\AxiomC{}
\RightLabel{$\top$}
\LeftLabel{Close2}
\UnaryInf$\ldots \fCenter B_1'' \cup C \cup \{ l\not=k,\, l\not=j,\, e_1=e_2 \}$
\RightLabel{$\top$}
\LeftLabel{Redminus2}
\UnaryInf$\ldots \fCenter B_1' \cup C \cup \{ l\not=k,\, l\not=j,\, e_1=e_2 \}$
\RightLabel{$\top$}
\LeftLabel{Redplus2}
\UnaryInf$\ldots \fCenter B_1' \cup C \cup \{ l\not=k,\, l\not=j \}$
\RightLabel{$\!\!\!\top$}
\LeftLabel{Disjunction2}
\AxiomC{\vdots}
\AxiomC{\vdots}
\TrinaryInf$\ldots \fCenter B_1' \cup C$
\LeftLabel{Propagate1}
\RightLabel{$a=wr(b,l,f_1)$}
\UnaryInf$A_1' \cup \{ f_1 \not= f_2,\, a=wr(b,l,f_1)\, \} \cup C \fCenter B_1 \cup C \cup \{ f_1 \not=f_2 \}$
\LeftLabel{Redminus1}
\RightLabel{$a=wr(b,l,f_1)$}
\UnaryInf$A_1 \cup \{\, f_1 \not= f_2,\, a=wr(b,l,f_1)\, \} \cup C \fCenter B_1 \cup C \cup \{ f_1 \not=f_2 \}$
\LeftLabel{Redplus1}
\RightLabel{$a=wr(b,l,f_1)$}
\UnaryInf$A_1 \cup \{\, f_1 \not= f_2\, \} \cup C \fCenter B_1 \cup C \cup \{ f_1 \not=f_2 \}$
\LeftLabel{Propagate1}
\RightLabel{$a=wr(b,l,f_1) \wedge f_1 \not= f_2$}
\UnaryInf$A_1 \cup \{\, f_1 \not= f_2\, \} \cup C \fCenter B_1 \cup C$
\LeftLabel{Redminus1}
\RightLabel{$a=wr(b,l,f_1) \wedge f_1 \not= f_2$}
\UnaryInf$A_1 \cup \{\, a\not=b,\, f_1 \not= f_2 \} \cup C \fCenter B_1 \cup C$
\LeftLabel{Redplus1}
\RightLabel{$a=wr(b,l,f_1) \wedge f_1 \not= f_2$}
\UnaryInf$A_1 \cup \{\, a\not=b\, \} \cup C \fCenter B_1 \cup C$
\LeftLabel{Define0*}
\RightLabel{$\varphi_2$}
\UnaryInf$A_1 \cup \{\, a\not=b\, \} \fCenter B_1$
\end{prooftree}
\medskip

where \\ 
$C \equiv \{ \diff(a,b)=l,\, rd(a,l)=f_1,\, rd(b,l)=f_2\, \}$ \\
$A_1' \equiv A_1 \setminus \{\, a=wr(b,i,d) \, \}$ \\
$B_1' \equiv B_1 \cup \{\, f_1\not=f_2,\, a=wr(b,l,f_1) \, \}$ \\
$B_1'' \equiv B_1' \setminus \{\, rd(a,j)=e_1\, \}$ \\
$\varphi_2 \equiv (a=wr(b,\diff(a,b),rd(a,\diff(a,b))) \wedge rd(a,\diff(a,b)) \not= rd(b,\diff(a,b)))$

\end{minipage}
}

\caption{Interpolant derivation for $\Pi_2$ using metarules. The
derivation is to be read bottom-up. The labels
for the rules are shown on the left, while the partial interpolants,
computed top-down, are shown on the right.}
\label{fig:derivation_pi_2}

\end{figure}

The final interpolant is computed by combining the interpolants for
$\Pi_1$ and $\Pi_2$ by means of (Disjunction1), yielding 
\begin{eqnarray*}
\varphi & \coincide & \varphi_1 \vee \varphi_2 \coincide \\
        & \coincide & (a=b \vee (a=wr(b,\diff(a,b),rd(a,\diff(a,b))) \wedge \\
	&           & \wedge\ rd(a, \diff(a,b))\neq rd(b,\diff(a,b)))
\end{eqnarray*}
which can be simplified to $\varphi \equiv (a=wr(b,\diff(a,b),rd(a,\diff(a,b))))$.

\section{Related work and Conclusions}
\label{sec:related}

There are two main lines of work in the literature which is relevant
for our paper: satisfiability procedures for variants and extensions
of the theory of arrays and interpolation methods related to the
theory of arrays.  Below, we discuss the works which are more closely
related to our approach in some details.

\subsection{Satisfiability}
\label{subsec:related-satisfiability}

Since its introduction by McCarthy in~\cite{mccarthy}, the theory of
arrays have received a lot of attention in automated theorem proving
and verification because of its importance in modelling fundamental
mechanisms of hardware and software systems such as memory read and
write operations.  For example, a lot of papers have been devoted to
design, prove correct, and build decision procedures for the
satisfiability problem of quantifier-free and selected classes of
quantified formulae in (various extensions of) the theory of arrays;
e.g.,~\cite{levitt,SBD+01,ARR,AMAI,kapurzarba,GD07,BNO+08b,GKF08,BB09b,dMB09}.
The interested reader is pointed to the `related work' sections
of~\cite{AMAI,dMB09} for a comprehensive overview.  Here, we notice
that many of them are based on instantiating the axioms of the theory
so that $rd$ and $wr$ can be considered as uninterpreted functions and
state-of-the-art procedures for the theory of equality can be used.
Notable exceptions are~\cite{ARR,levitt,SBD+01} where techniques based
on rewriting or constraint solving are used.

In~\cite{ARR}, the standard superposition
calculus~\cite{superposition} is proven to terminate on the union of
the theory of arrays and a set of ground literals; thereby, providing
a decision procedure for the quantifier-free satisfiability problem
because of the refutation completeness of the calculus.  (The
efficiency of the approach is explored in~\cite{Bonacina}.)  While the
saturation (roughly, the exhaustive application of the rules of the
superposition calculus) can be seen as a generalization of completion
where clauses, and not only equalities, are handled, our Gaussian
completion\footnote{The Gauss elimination procedure for systems of
  linear equalities has been lifted to elementary theories
  in~\cite{gaussBG} and, since the theory of arrays is close to being
  Gaussian~\cite{gaussRB}, we show that `Gaussian-like' steps can be
  exploited during completion phase.} has some distinctive features.
In fact, while the three critical pairs (C3), (C4), and (C5) in
Section~\ref{subsec:gaussian} can be regarded as instances of the
inference rules of a superposition calculus (see~\cite{ARR} for
details), the critical pairs (C1) and (C2), exploiting the
equivalences in Figure~\ref{fig:key-lemmas}, are impossible to recast
in any standard completion procedure (see, e.g.,~\cite{rewriting}).
In fact, the way in which the critical pairs (C1) and (C2) are
eliminated involves the addition of equalities containing $rd$'s (in
order to constrain the values stored at certain locations in the
arrays mentioned in the rules of the critical pair) besides the
replacement of one or both the parent rewrite rules by an equality.
Only in this way, we were able to eliminate badly orientable rules.
It seems difficult to adapt the approach in~\cite{ARR} to the problem
under consideration mainly because of the chosen order $>$ over terms.
In fact, we orient the equality $a\to wr(b,i,e)$ from left to right if
$a>b$, and use the equivalences in Figure~\ref{fig:key-lemmas} when
$b>a$ (or $a$ and $b$ are identical).  This allows us to eliminate all
critical pairs with rules \eqref{eq:r1}--\eqref{eq:r4} in
Definition~\ref{def:normal} since such rules contain just one
variable of sort \ARRAY and, trivially, no critical pairs involving
the variable should be considered.  If we choose the other way of
orienting the equalities of the form $a=wr(b,i,e)$, several critical
pairs would arise.  Although the completion of these pairs terminate
under suitable assumptions (as shown in~\cite{ARR}), this creates
serious problems when considering the computation of interpolants.

In~\cite{levitt}, a satisfiability procedure for the theory of arrays
with extensionality is designed so as to be easily combined with other
procedures by the Shostak combination method (see, e.g.,~\cite{RRT04}).
Two interface functionalities are required by the Shostak combination
method: (i) normalizing terms and (ii) solving equalities.  We
consider each activity in details.
\begin{enumerate}[(i)]
\item In Chapter 5 of~\cite{levitt}, a canonical form for
  terms built out by using a single $rd$ or several $wr$'s is defined
  by using a simplification ordering.  The canonical terms are similar
  to those occurring in a \nform\ constraint according to
  Definition~\ref{def:normal} above.
  A major difference is the use of if-then-else's to normalize
  read-terms in~\cite{levitt} while our procedure does not use them
  because item (i) of Definition~\ref{def:normal} implies that any two
  indexes in a constraint in normal form are known to be distinct.
  This choice makes the proof of the correctness of our procedure much
  easier with respect to the argument for the correctness proposed
  in~\cite{levitt} which ``\emph{has proved elusive to the authors}''
  of~\cite{SBD+01}.  So called `lazy' SMT solvers, based on the
  integration of a SAT solver and a satisfiability procedures for
  conjunction of literals, seem to be able to easily implement the
  case-splitting required to derive a complete partition by resorting
  to the available SAT solver as explained, e.g., in~\cite{BBC+05c}.
\item To compare with the activity of solving equalities
  in~\cite{levitt}, let us preliminarily observe that the logical
  equivalences in Figure~\ref{fig:key-lemmas} can be considered as
  rewrite rules (either from left to right or viceversa) that help us
  replace badly orientable equalities (recall the definition at the
  beginning of Section~\ref{sec:interpolation}) with equalities which
  are oriented from left to right.  This is precisely how the
  equivalences in Figure~\ref{fig:key-lemmas} are used in the Gaussian
  completion procedure (of Section~\ref{subsec:gaussian}) to eliminate
  critical pairs.  Similarly, in order to provide one of the basic
  functionalities required by the Shostak combination framework,
  \cite{levitt} designs a solver for equalities involving $wr$
  operations.  For example, the procedure in~\cite{levitt} allows one
  to solve the equality $a=wr(b,i,e)$ for $b$.  We can adapt our
  procedure (in particular, by using the equivalences \textbf{Symm}
  and \textbf{Refl} of Figure~\ref{fig:key-lemmas})
  to do the same.  The main difference is that
  our normalization is done off-line, i.e.\ the signature is fixed
  since all terms appearing in the constraint are given, while the
  procedure in~\cite{levitt} must be on-line since is to be integrated
  in a Shostak combination algorithm which requires that to process
  equalities one at a time, as soon as they become available.  Because
  of this, the completion algorithm can be simplified (since there is
  no need to compute intermediate normal forms) and standard
  techniques to show its termination can be used.  In contrast,
  ~\cite{levitt} gives only a brief sketch of the termination of his
  procedure.  For a more comprehensive comparison of on-line and
  off-line completion algorithms revisiting the Shostak congruence closure
  algorithm, the reader is pointed to~\cite{kapur,tiwari}.
\end{enumerate}
The procedure in~\cite{SBD+01} share with~\cite{levitt} and ours the
key activity of solving equalities.  The main difference is that no
canonical forms for terms or constraints are defined in~\cite{SBD+01};
rather a special form of equality over arrays is introduced, called
partial equality, which compares the content of two arrays only at a
(finite) set of indexes.  Formally, this is defined as follows: $a =_I
b$ iff for every index $i$ not in the set $I$, the content of $a$ at
$i$ is equal to that of $b$ at the same index $I$.  Thus, an equality
of the form $wr(a,i,e)=b$ can be rewritten as $a=_{\{i\}}b \wedge
rd(b,i)=e$.  The key insight of~\cite{SBD+01} is that it is possible
to eliminate all $wr$'s, so that arrays can be considered as
uninterpreted functions and $rd$ as function application, and a
slightly modified congruence closure (to cope with partial equality)
can be used to check satisfiability.  While no standard rewriting
techniques are used in~\cite{SBD+01}, it is interesting to notice that
two arrays $a$ and $b$ are cardinality dependent iff there exists a
finite set $I$ of indexes such that $a=_I b$.  
We do not introduce
a new predicate symbol and use it in designing a satisfiability
procedure, 
however we nevertheless exploit this notion and its preservation through embeddings 
(see Lemma~\ref{lem:dependency}) during our semantic interpolation proofs.

\subsection{Interpolation}

After McMillan's seminal work on interpolation for model
checking~\cite{McM03,McM04a},
 several
papers~\cite{jhala,McM04b,stokkermans,YM05,KMZ06,RS07,CGS08,voronkov,lynch,CGS10,BKR+10}
appeared whose aim was to design techniques for the efficient
computation of interpolants in first-order theories of interest for
verification, mainly uninterpreted function symbols, fragments of
Linear Arithmetic, or their combination.  
An interpolating
theorem prover is described in~\cite{McM05}, where a sequent-like
calculus is used to derive interpolants from proofs in propositional
logic, equality with uninterpreted functions, linear rational
arithmetic, and their combinations.  The method described
in~\cite{YM05} proposes a framework suitable for lazy SMT-solvers, in
which the theory solver is required to derive partial interpolants for
each theory lemmata it produces. The global interpolant can then be
computed at the propositional level. The paper also illustrates a
method to derive interpolants in a Nelson-Oppen combination procedure,
under certain restrictions on the theories to combine.  More recently,
in~\cite{CGS10} the ideas of~\cite{YM05} are adapted to cope with
state-of-the-art SMT-solving strategies for combinations of the
theories of uninterpreted functions and a fragment of Linear
Arithmetic (called difference logic).  In~\cite{KMZ06}, a method to
compute interpolants in data structures theories, such as sets and
arrays (with extensionality), by axiom instantiation and interpolant
computation in the theory of uninterpreted functions is described.  It
is also shown that the theory of arrays with extensionality does not
admit quantifier-free interpolation.
The ``split'' prover in~\cite{jhala} applies a sequent calculus for
the synthesis of interpolants along the lines of that in~\cite{McM05}
and is tuned for predicate abstraction~\cite{graf}.  In particular,
the method is shown to be complete in the sense that the computed
interpolants are guaranteed to provide the ``right'' level abstraction
to prove a certain property, if one exists.  The ``split'' prover can
handle a combination of theories among which also the theory of arrays
without extensionality is considered.  In~\cite{jhala}, it is pointed
out that the theory of arrays poses serious problems in deriving
quantifier-free interpolants because it entails an infinite set of
quantifier-free formulae, which is indeed problematic when
interpolants are to be used for predicate abstraction.  To overcome
the problem, \cite{jhala} suggests to constrain array valued terms to
occur in equalities of the form $a=wr(a,I,E)$ in the notation of this
paper.  It is observed that this corresponds to the way in which
arrays are used in imperative programs.  Further limitations are
imposed on the symbols in the equalities in order to obtain a complete
predicate abstraction procedure.  In~\cite{jhala-array}, the method
described in~\cite{jhala} is specialized to apply CEGAR
techniques~\cite{CGJ+00} for the verification of properties of
programs manipulating arrays.  The method of~\cite{jhala} is extended
to cope with range predicates which allow one to describe unbounded
array segments which permit to formalize typical programming idioms of
arrays, yielding property-sensitive abstractions.
In~\cite{stokkermans}, it is shown how to extend satisfiability
procedures based on axiom instantiation to compute interpolants.
However, the theory of arrays is not considered.  In~\cite{RS07}, the
approach of~\cite{stokkermans} is specialized to compute interpolants
in the combination of Linear Rational Arithmetic and the theory of
uninterpreted function symbols; again, the theory of arrays is not
considered.  A method for deriving interpolants in the theory of
equality with uninterpreted functions is also given in~\cite{KFG+09}
by extending a congruence closure algorithm.  In~\cite{voronkov}, a
method to derive quantified invariants for programs manipulating
arrays and integer variables is described.  A resolution-based prover
is used to handle an \emph{ad hoc} axiomatization of arrays by using
predicates.  Neither McCarthy's theory of arrays nor one of its
extensions are considered in~\cite{voronkov}.  The invariant synthesis
method is based on the computation of interpolants derived from the
proofs of the resolution-based prover and constraint solving
techniques to handle the arithmetic part of the problem.  The
resulting interpolants may contain even alternation of quantifiers.

Latest research on interpolating procedures has been focusing on
(extensions of) Linear Integer Arithmetic.  An interpolating procedure
for linear Diophantine equalities is outlined in~\cite{JCG09}. A
procedure for full Linear Integer Arithmetic based on a sequent
calculus can be found in~\cite{BKR+10}.  In~\cite{BKR+10b}, the
procedure in~\cite{BKR+10} is extended to cope with the theory of
arrays without extensionality by axiom instantiation and interpolation
in the combination of Presburger Arithmetic and uninterpreted
function.  Quantifiers can occur in the interpolants returned by the
procedure.  Recently~\cite{frocos11}, we have proposed a
quantifier-free interpolation solver for \AXDIFF when combined with
integer difference logic over indexes.

\subsection{Conclusions and Future Work} 


We believe that the procedure proposed in this paper is a significant
step forward to make model-checking more widely applicable to programs
whose properties depend crucially on the manipulations of arrays.  
To the best of our knowledge, in fact, our interpolation procedure is
the first to compute quantifier-free interpolants for a natural
variant of the theory of arrays with extensionality obtained by
replacing the extensionality axiom with its Skolemization.  This
variant is `natural' in the sense that it is sufficient to detect
unsatisfiability of formulae as it is usually the case in standard
model checking methods for infinite state systems.  

Despite the work reported in this paper is a significant step forward
in widening the scope of applicability of interpolation in model
checking of array manipulating programs, we discuss some interesting
directions for further work.

The implementation of the interpolating procedure proposed here is
crucial for showing the practical viability of our approach.  In this
respect, the first step is to implement the satisfiability solver in
Section~\ref{sec:solver}.  Recall that this requires guessing, a
pre-processing phase, and Gaussian completion phase.  Guessing, as
already observed in Section~\ref{subsec:related-satisfiability} item
\textbf{(i)} when discussing the relationship with the solver
of~\cite{levitt}, can be implemented by adapting the mechanism to
handle arrangements when combining satisfiability procedures in the
Delayed Theory Combination approach of~\cite{BBC+05c}.  The main
advantage of this approach is to use state-of-the-art SAT techniques
to efficiently enumerate all possible partitions of indexes.  The
pre-processing phase can be implemented by using the data structures
and basic expression manipulating procedures available in many
state-of-the-art SMT solvers.  The Gaussian completion phase requires
more effort but it can adapt and reuse well-known techniques developed
in rewriting for completion procedures (see, e.g.,~\cite{rewriting}).
The second step to build the interpolating procedure of
Section~\ref{sec:isolver} is to implement the interpolating metarules
of Table~\ref{tab:metarules}.  This is relatively simple and does not
require much ingenuity and can be done on top of the existing
infrastructure for proof generation that is available in many
state-of-the-art SMT solvers.

We are currently developing an implementation of the procedure
presented here in the SMT-solver OpenSMT~\cite{BPS+10}.  Preliminary
experiments are encouraging although a more extensive experimental
evaluation is needed.  
In fact, it is well-known that the convergence of interpolation based
model checking procedures crucially depends on the ``quality'' of the
computed interpolants.  There have been attempts (see,
e.g.,~\cite{jhala,quantified-mcmillan}) to build interpolating
procedures that return ``high quality'' interpolants that guarantee
the convergence of model checking for valid properties.  Recently, it
has been observed~\cite{interp-strength,mcmillan-z3-interpolation}
that a certain degree of flexibility for tuning the computation of
interpolants in interpolation procedures would be desirable to
facilate their integration in model checking.  In this respect, it
would be particularly interesting to investigate how the order in
which the interpolating metarules of Table~\ref{tab:metarules} are
applied, particularly those on \abcommon{} terms, may influence the
``quality'' of the interpolants.  An interesting alternative to
investigate the flexibility of generating interpolants (suggested
in~\cite{mcmillan-z3-interpolation}) would be to use the procedure
presented here in the framework for computing quantified interpolants
of~\cite{mcmillan-z3-interpolation}.

Finally, there are two more interesting points that deserve further
investigations.  First, it would be interesting to study the size of
the interpolating metarules refutations and compare them with
interpolating procedures based on a proof calculus.  The preliminary
experiments with our implementation of the procedure in Open SMT show
that our refutations are quite compact but a more systematic
comparison with available procedures based on a proof calculus,
e.g.,~\cite{McM05} is needed to clarify this issue.  Second, since in
model checking it is useful to compute interpolants for several
partitions of the same (unsatisfiable) formula, it would be
interesting to design a method that permit the partial reuse of the
interpolants returned for a partition to compute the interpolant for
the next one so as to permit reuse and avoid degradation of
performances due to partial recomputation of parts of interpolating
metarules refutation.  In this respect, it seems possible to adapt
techniques developed for computing chains of interpolants
in~\cite{iprincess2011}.

\ \\

\noindent \emph{Acknowledgements}. 
We wish to thank two anonymous referees for their comments on a draft
of this paper and an anonymous referee of RTA'11 for the criticisms
that helped improving the presentation.

The work of the third author was partially supported by the
``Automated Security Analysis of Identity and Access Management
Systems (SIAM)'' project funded by Provincia Autonoma di Trento in the
context of the ``team 2009 - Incoming'' COFUND action of the European
Commission (FP7) and the FP7-ICT-2007-1 Project no.~216471.

\bibliographystyle{plain}
\bibliography{brutt,ghila,ranis}

\newpage

\appendix

\newpage

\section{Proof of Theorem~\ref{thm:interpolation-amalgamation}}
\label{app:appknown}

\ \\

\noindent\textbf{Theorem~\ref{thm:interpolation-amalgamation}\cite{amalgam}}
\emph{ 
  Let $T$ be universal. Then, $T$ admits quantifier-free
  interpolation iff $T$ has the amalgamation property.  
}
\vskip 2mm

\begin{proof}
\emph{Suppose first that $T$ has amalgamation}; let $A, B$ be
quantifier-free formulae such that $A\wedge B$ is not
$T$-satisfiable. Let us replace variables with free constants in $A,
B$; let us call $\Sigma^A$ the signature $\Sigma$ of $T$ expanded with the
free constants from $A$ and $\Sigma^B$ the signature $\Sigma$ expanded
with the free constants from $B$ (we put $\Sigma^C \coincide \Sigma^A\cap
\Sigma^B$). For reductio, suppose that there is no ground formula $C$
such that: (a) $A$ $T$-entails $C$; (b) $C\wedge B$ is
$T$-unsatisfiable; (c) only free constants from $\Sigma^C$ occur in
$C$.

As a first step, we build a maximal $T$-consistent set $\Gamma$ of
ground $\Sigma^A$-formulae and a maximal $T$-consistent set
$\Delta$ of ground $\Sigma^B$-formulae such that $A\in
\Gamma$, $B\in \Delta$, and $\Gamma\cap \Sigma^C = \Delta \cap
\Sigma^C$.\footnote{ By abuse, we use $\Sigma^C$ to indicate not only
  the signature $\Sigma^C$ but also the set of formulae in the
  signature $\Sigma^C$.  } For simplicity\footnote{ This is just to
  avoid a (straightforward indeed) transfinite induction argument.}
let us assume that $\Sigma$ is at most countable, so that we can fix
two enumerations
$$
A_1, A_2, \dots \qquad B_1, B_2, \dots
$$ 
of ground $\Sigma^A$- and $\Sigma^B$-formulae,
respectively. We build inductively $\Gamma_n, \Delta_n$ such that for
every $n$ (i) $\Gamma_n$ contains either $A_n$ or $\neg A_n$; (ii)
$\Delta_n$ contains either $B_n$ or $\neg B_n$; (iii) there is no
ground $\Sigma^C$-formula $C$ such that $\Gamma_n \cup
\{\neg C\}$ and $\Delta_n\cup \{C\}$ are not $T$-consistent.  Once
this is done, we can get our $\Gamma, \Delta$ as $\Gamma:=\bigcup
\Gamma_n$ and $\Delta := \bigcup \Delta_n$.

We let $\Gamma_0$ be $\{ A\}$ and $\Delta_0$ be $\{B\}$ (notice that
(iii) holds by (a)-(b)-(c) above). To build $\Gamma_{n+1}$ we have two
possibilities, namely $\Gamma_n \cup \{ A_n\}$ and $\Gamma_n \cup \{
\neg A_n\}$. Suppose they are both unsuitable because there are $C_1,
C_2\in  \Sigma^C$ such that the sets
$$
\Gamma_n \cup \{ A_n,\neg C_1\}, \quad \Delta_n\cup \{C_1\}, \quad \Gamma_n \cup \{\neg A_n, \neg C_2\}, \quad \Delta_n\cup \{C_2\}
$$
are all $T$-inconsistent. If we put $C:= C_1\vee C_2$, we get that $\Gamma_n \cup \{\neg C\}$ and $\Delta_n\cup \{C\}$ are not $T$-consistent, contrary to induction hypothesis. A similar argument
shows that we can also build $\Delta_n$.

Let now $\cM_1$ be a model of $\Gamma$ and $\cM_2$ be a model of
$\Delta$. Consider the substructures $\cN_1, \cN_2$ of $\cM_1, \cM_2$
generated by the interpretations of the constants from $\Sigma^C$:
since the related diagrams are the same
(because  $\Gamma\cap \Sigma^C = \Delta \cap \Sigma^C$), we have that $\cN_1$ and $\cN_2$ are $\Sigma_C$-isomorphic.
 Up to renaming, we can suppose that $\cN_1$ and $\cN_2$ are just the same substructure (let us  call it $\cN$ for short). Since the theory $T$ is universal and truth of universal sentences is preserved by substructures,
we have that  $\cN$ is a model of $T$. By the amalgamation property, there is a $T$-amalgam $\cM$ of $\cM_1$ and $\cM_2$ over $\cN$. Now $A, B$ are ground formulae true in $\cM_1$ and $\cM_2$, respectively,
hence they are both true in $\cM$, which is impossible because $A\wedge B$ was assumed to be $T$-inconsistent.

\vskip 2mm
\emph{Suppose now that $T$ has quantifier free interpolants}. Take two models $\cM_1=(M_1, {\mathcal I}_1)$ and $ \cM_2=(M_2, {\mathcal I}_2)$ of $T$ sharing a substructure $\cN=(N, {\mathcal J})$. In order
to show that a $T$-amalgam of $\cM_1, \cM_2$ over $\cN$ exists, it is sufficient
(by Robinson Diagram Lemma~\ref{lem:robinson}) to show that $\delta_{\cM_1}(M_1)\cup \delta_{\cM_2}(M_2)$ is $T$-consistent. If it is not, by the compactness theorem of first order logic, there exist a
$\Sigma\cup {M_1}$-ground sentence $A$ and a $\Sigma\cup {M_2}$-ground sentence $B$ such that (i) $A\wedge B$ is $T$-inconsistent; (ii) $A$ is a conjunction of literals from  $\delta_{\cM_1}(M_1)$; (iii)
$B$ is a conjunction of literals from  $\delta_{\cM_2}(M_2)$.
By the existence of quantifier-free interpolants, taking free constants instead of variables, we get that there exists a ground $\Sigma\cup N$-sentence $C$ such that $A$ $T$-entails $C$ and $B\wedge C$ is $T$-inconsistent. The former fact yields that $C$ is true in $\cM_1$ and hence also
in $\cN$ and in $\cM_2$,
because $C$
is ground.
 However, the fact that $C$ is true in $\cM_2$ contradicts the fact that $B\wedge C$ is $T$-inconsistent.
\end{proof}

\end{document}